\documentclass[final,twoside,11pt]{entics} 
\usepackage{enticsmacro}

\usepackage[noadjust]{cite}
\usepackage{amsmath,amssymb,amsfonts}
\usepackage{bbm,cmll}
\usepackage{enumitem}
\usepackage{hyperref}
\usepackage[inference]{semantic} 
\usepackage{proof} 
\usepackage{xifthen}
\usepackage{algorithmic}
\usepackage{graphicx}
\usepackage{textcomp}
\usepackage{xcolor}
\usepackage{subcaption}
\usepackage{parskip}
\usepackage{thmtools}
\usepackage{thm-restate}

\newtheorem{notn}[thm]{Notation}

\providecommand*{\ifempty}[3]{\ifthenelse{\isempty{#1}}{#2}{#3}}

\newcommand{\eg}{\textit{e.g.}}
\newcommand{\ie}{\textit{i.e.}}
\newcommand{\wrt}{w.r.t.}

\renewcommand{\e}{\varepsilon}
\newcommand{\un}{\mathbbm{1}}

\newcommand{\D}{\mathcal{D}}
\newcommand{\R}{\mathcal{R}}
\newcommand{\M}{m}
\newcommand{\N}{\mathcal{N}}
\newcommand{\X}{\mathcal{X}}
\newcommand{\Y}{\mathcal{Y}}
\newcommand{\term}{\mathbb{T}}

\newcommand{\lol}{\multimap}
\newcommand{\lollol}{\multimapboth}
\newcommand{\ot}{\otimes}
\newcommand{\T}{\term}

\newcommand{\naturals}{\mathbb{N}}

\newcommand{\reals}{\mathbb{R}}
\newcommand{\preals}{[0,\infty)}

\newcommand{\extreals}{[0,\infty]}

\newcommand{\Prop}[1][]{\mathbb{P}_{#1}}
\newcommand{\LLQ}{\text{LLQ}}
\newcommand{\LL}{\mathbb{L}}
\newcommand{\LLun}{\LL_{\un}}
\newcommand{\LLs}{\LL_{\un}^*}
\newcommand{\Half}{\textstyle\frac{1}{2}}
\newcommand{\coloneq}{\mathrel{\mathop:}=}
\makeatletter
\newcommand{\dotdiv}{\mathbin{\text{\@dotminus}}}

\newcommand{\@dotminus}{%
  \ooalign{\hidewidth\raise1ex\hbox{.}\hidewidth\cr$\m@th-$\cr}%
}
\makeatother

\newcommand{\infrule}[3][]{\infer[{\ifempty{#1}{}{(\textsc{#1})\;}}]{#3}{#2}}
\newcommand{\doubleinfrule}[3][]{\infer=[{\ifempty{#1}{}{(\textsc{#1})\;}}]{#3}{#2}}

\newcommand{\gcut}[1]{}

\usepackage{tikz}
\usetikzlibrary{calc,matrix,arrows} 
  \tikzset{
    commutative diagram/.style 2 args={
    	matrix of math nodes, row sep=#1,column sep=#2,
	text height=1.5ex, text depth=0.25ex},
    commutative diagram/.default={1cm}{1cm}
    }
  \tikzset{    
    skip loop/.style n args={3}{to path={-- ++(0,#1) -| node[pos=0.25,#2] {#3} (\tikztotarget)}},
    cross line/.style={preaction={draw=white, -, line width=6pt}}
  }

\tikzset{
	treenode/.style = {shape=rectangle, rounded corners,
		draw, align=center,
		top color=white, bottom color=blue!20},
	root/.style     = {treenode, font=\ttfamily\scriptsize, bottom color=red!30},
	env/.style      = {treenode, font=\ttfamily\scriptsize},
	dummy/.style    = {circle,draw}
}

\def\conf{MFPS 2023} 	
\def\myconf{mfps39}
\volume{3}			
\def\volu{3}			
\def\papno{3}			
\def\lastname{Bacci, Mardare, Panangaden, Plotkin} 
\def\runtitle{Propositional Logics for the Lawvere Quantale} 
\def\mydoi{12292}
\def\copynames{G. Bacci, R. Mardare, P. Panangaden, }    
\begin{document}
\begin{frontmatter}
 \title{Propositional Logics for the Lawvere Quantale}				
  \author{Giorgio Bacci\thanksref{a}\thanksref{emailA}}	                 
  \author{Radu Mardare\thanksref{b}\thanksref{emailB}}		
  \author{Prakash Panangaden\thanksref{c}\thanksref{emailC}}
  \author{Gordon Plotkin\thanksref{d}\thanksref{emailD}}		
   \address[a]{Department of Computer Science, Aalborg University, 			
    Aalborg, Denmark}  							
   \thanks[emailA]{Email: \href{mailto:grbacci@cs.aau.dk} {\texttt{\normalshape
        grbacci@cs.aau.dk}}} 
  \address[b]{Computer and Information Sciences, University of Strathclyde, 
    Glasgow, Scotland} 
  \thanks[emailB]{Email:  \href{mailto:r.mardare@strath.ac.uk} {\texttt{\normalshape
        r.mardare@strath.ac.uk}}}
  \address[c]{School of Computer Science, McGill University and Mila,
    Montr\'eal, Canada} 
  \thanks[emailC]{Email:  \href{mailto:prakash@cs.mcgill.ca} {\texttt{\normalshape
        prakash@cs.mcgill.ca}}}
  \address[d]{LFCS, School of Informatics, University of Edinburgh, 
    Edinburgh, Scotland} 
  \thanks[emailD]{Email:  \href{mailto:gdp@inf.ed.ac.uk} {\texttt{\normalshape
        gdp@inf.ed.ac.uk}}}
\begin{abstract} 
Lawvere showed that generalised metric spaces are categories enriched over $[0, \infty]$,
the quantale of the positive extended reals. The statement of enrichment is a quantitative
analogue of being a preorder. Towards seeking a logic for quantitative metric reasoning,
we investigate three $[0,\infty]$-valued propositional logics over the Lawvere quantale. 
The basic logical connectives shared by all three logics are those that can be interpreted
in any quantale, viz finite conjunctions and disjunctions, tensor (addition for the
Lawvere quantale) and linear implication (here a truncated subtraction); to these we add,
in turn, the constant $\un$ to express integer values, and scalar multiplication by a
non-negative real to express general affine combinations.  Quantitative equational logic
can be interpreted in the third logic if we allow inference systems instead of axiomatic
systems.


For each of these logics we develop a natural deduction system which we prove to be
decidably complete w.r.t.\ the quantale-valued semantics.  The heart of the completeness
proof makes use of the Motzkin transposition theorem.  Consistency is also decidable; the
proof makes use of Fourier-Motzkin elimination of linear inequalities.
Our logics are variants on $[0,1]$-logics, via the exponential map $e^{-x}$: the first is
equivalent to Goguen's product logic~\cite{Goguen74}, the second is a 
version of product logic with  a fixed choice of
constants, the third adds exponentiation. The first \text{is} known; the 
second and third are apparently novel. 
The third is natural in the additive $[0,\infty]$-context, but not the multiplicative
$[0,1]$-context.  Our proofs are novel in all cases, making uses of linear algebraic
results; these results are commonly regarded informally as a kind of completeness result;
here we put that view to good use.
Strong completeness does not hold in general, even (as is known) for theories over
finitely-many propositional variables; indeed even an approximate form of strong
completeness in the sense of Pavelka or Ben Yaacov---provability up to arbitrary
precision---does not hold.  However, we can show it for theories axiomatized by a (not
necessarily finite) set of judgements in \emph{normal form} over a finite set of
propositional variables when we restrict to models that do not map variables to $\infty$;
the proof uses Hurwicz's general form of the Farkas' Lemma.
\end{abstract}
\end{frontmatter}

\section{Introduction}

Real-valued logics have been experiencing a recent resurgence of interest
because of probabilistic~\cite{Kozen81,Kozen85,Baier08} and metric
reasoning~\cite{Panangaden09} and applications such as neurosymbolic 
reasoning in machine learning~\cite{Fagin94,Fagin2020,sen2022neuro}.   
They are generally fuzzy logics~\cite{HajekBook} often interpreted over $[0,1]$---for example {\L}ukasiewicz logic (\cite{Lukasiewicz30,Yaacov13,yaacov10}).
%
In~\cite{Lawvere73}, Lawvere showed that
generalised metric spaces are categories enriched over $[0, \infty]$, the quantale of the
positive extended reals. The statement of enrichment is a quantitative analogue of being a
preorder. One can view {\L}ukasiewicz logic as the propositional logic of 
a quantale over $[0,1]$, see~\cite{diNola2011}. Also there are interpretations of
(variants of) linear logic in quantales~\cite{Yetter90}. 

In this paper, we propose studying logic over Lawvere's quantale, and we begin such a study with propositional logic.
An argument that points towards this quantale comes
also from the literature on quantitative algebras~\cite{Mardare16}, where the development
of quantitative equational logic is based on a family of binary predicates "$ =_\e$" for every $\e\geq
0$.  These are used in quantitative equations, such as $s=_\e t$, to encode the fact that
the distance between the interpretation of the terms $s$ and $t$ in a metric
space is at most $\e$. Once this quantitative information is embedded in the predicate
"$=_\e$", quantitative equations become Boolean statements, i.e., \textit{true} or
\textit{false} in a model. 
An alternative way to tackle this issue is to consider only one "$=$" predicate that is valued
in the Lawvere quantale. This requires exchanging the classical ``Boolean core'' of 
quantitative equational logic with a many-valued logic interpreted over Lawvere's
quantale.

 
In this paper, we investigate  basic concepts and proof systems for such
propositional logics. We consider three closely related logics,  built up in
stages. 
The basic logical
connectives, shared by all three logics, are those that can be interpreted in any quantale,
viz finite conjunctions and disjunctions, tensor (addition for the Lawvere quantale) and
linear implication (here a truncated subtraction). To these we add, in turn, the constant
$\un$ to express integer values, and  scalar multiplication by a non-negative 
 real to express
general non-negative affine combinations. 
Quantitative equational logic can be interpreted in the third logic once we extend the provability principles and instead of only looking to theories defined by axiomatic systems, we also consider theories closed under systems of inferences. 

 Our logics are variants of $[0,1]$-valued logics, via the exponential map $e^{-x}$: the
 first is equivalent to product logic~\cite{H96, H06}, the second is a 
 version of product logic with a fixed choice of constants~\cite{Sc06,H95}, the third adds
 exponentiation.  The first \text{is} known; the second and third are apparently novel.
 The third is natural in the additive $[0,\infty]$-context, but not the multiplicative
 $[0,1]$-context, where multiplication by a constant becomes
 exponentiation. Multiplication by a constant can be viewed as a graded modality; however
 the literature on product logic modalities seems to consider only Kripke
 models~\cite{VE17}.
 
 Our proofs are novel in all cases. The proofs used in the fuzzy logic literature, e.g.,
 in~\cite{H96,Sc06}, employ the methods of algebraic logic, and for product logic in
 particular, use results of ordered Abelian group theory.  Our proofs are arguably
 simpler. They consist of a syntactic reduction to a normal form using only affine
 combinations of propositional variables, followed by a novel application of theorems from
 linear algebra, such as the Farkas' Lemma. It has been long noted in the
 literature~\cite{MatousekBook} that the Farkas' Lemma, and related results, can be
 thought of as completeness theorems. Here they are literally seen as such, and used to
 prove general completeness for our propositional logics.

 For each of these logics we develop a natural deduction system, and prove it complete
 relative to interpretations in the Lawvere quantale.  The proof of completeness uses a
 normalisation technique which replaces a sequent $\varphi\vdash\psi$ that is not in
 normal form, with a finite set of sequents in normal form. The main normal form is a
 sequent where the formulas $\varphi$, $\psi$ are tensors
 $r_1*p_1 \otimes \ldots \otimes r_n*p_n\otimes r$ of propositional variables multiplied
 by positive reals, and a scalar. Semantically, these are exactly affine linear
 combinations of the variables. Leaving aside the details, this reduces the problem of
 proving that a given sequent follows from a given finite set of sequents, to the problem
 of proving that a given affine inequality is a consequence of a given finite set of
 affine inequalities.
 This is exactly the province of (variants of) the Farkas' Lemma~\cite{Farkas1902} and the
 Motzkin transposition theorem~\cite{Motzkin51}. These show that when such a consequence
 holds for all values of the variables, then the given affine inequality is a linear
 combination of the given set of inequalities (an integer variant of
 Motzkin~\cite{Stoer2012} is used when the reals are integers). It is in this way that the
 Farkas' Lemma and its relatives can be viewed as completeness results (rather than, as in
 another--very similar--view, as a dichotomy result that either there is a linear
 combination or there is a counterexample).


 As reduction to sets of normal forms and the Farkas' Lemma (and related) are effective,
 it follows that satisfiability is decidable. We can also decide consistency via the
 reduction to normal form, now making use of Fourier-Motzkin
 elimination~\cite{Fourier1826,Williams1986}.  By making these reductions efficient one
 can further show that consequence is co-NP complete and that consistency is NP-complete.
 
 
Strong completeness does not
hold in general, even for theories over finitely-many propositional variables; indeed even
an approximate form of strong completeness in the sense of Pavelka or Ben Yaacov
\cite{EG00, P79, Yaacov13,yaacov10}---provability up to arbitrary precision---does not hold. 
However, we can show it does hold for such theories that admit a, not necessarily finite,
set of axioms in \emph{normal form} over a \emph{finite} set of propositional variables
when we restrict to models that do not map variables to $\infty$; the proof uses Hurwicz's
general form~\cite{Hurwicz2014} of the Farkas' Lemma.

\paragraph*{Synopsis}
In Section~\ref{sec:prelim}, we present some preliminaries and  basic notation. In Section~\ref{sec:LLQ}, we introduce the three quantale-valued logics with their syntax and semantics. Section~\ref{sec:proofSys} is dedicated to natural deduction systems, one for each of the three logics. It contains soundness statements and a series of results regarding provability in these logics. A subsection here is dedicated to deduction theorems, with a detailed discussion about the failure of the classic deduction theorem. 
In Section~\ref{sec:theories}, we collect a series of model-theoretic results and we introduce the concept of a diagrammatic theory that is central in the model theory of these logics. Section~\ref{sec:normalform} is dedicated to normalization; it contains a brief description of an algorithm that computes the normal form of a sequent. Due to space restrictions, 
the algorithm is only explained for a couple of key cases.
 Section~\ref{sec:completeness} is dedicated to completeness and incompleteness results for all these logics. In Section~\ref{sec:infSystems}, we show how we can encode  
quantitative equational logic in our logic. Quantitative equational reasoning requires more than just an axiomatic system; it needs a system of inferences that does not produce one axiomatic theory, but a set of axiomatic theories.
\section{Preliminaries and Notation}
\label{sec:prelim}

A \emph{quantale} is a complete lattice with a binary, associative operation $\ot$ (the \emph{tensor})
preserving all joins (and so with both $a \ot -$ and $- \ot a$ having  adjoints).
A quantale is called \emph{commutative} whenever its tensor is; it is called \emph{unital} if there
is an element $1$, the unit, such that $1 \ot a = a = a \ot 1$, for all $a$; and it is called
\emph{integral} if the unit is the top element.
For commutative quantales, where the adjoints coincide, we denote the  adjoint to $a \ot -$ by $a \lol -$,
which is characterised by
\begin{equation*}
 a \ot b \leq c  \Longleftrightarrow  b \leq a \lol c \,.
\end{equation*}


The complete lattice $\extreals$ ordered by the ``greater or equal'' relation $\geq$ with extended sum as tensor, is known as the \emph{Lawvere quantale} (or metric quantale).
In this paper, we mainly work with the Lawvere quantale, so it is convenient to have an
explicit characterisation of its basic operations. Join and meet are  
$\inf$ and $\sup$, respectively, $\infty$ is the bottom element and 
$0$ the top. For $r, s \in \extreals$, we define truncated subtraction as
\begin{equation*}
  r \dotdiv s = 
  \begin{cases}
     0 		&\text{if $r \leq s$} \\
     r - s 	&\text{if $r > s$ and $r \neq \infty$} \\
     \infty 	&\text{if $r = \infty$ and $s \neq \infty$} \,.
  \end{cases}
\end{equation*}
Then, the right adjoint $s \lol r$ is just $r \dotdiv s$ (note that the order of terms is inverted).
We remark that the continuous t-norms of fuzzy logic are exactly the integral commutative
quantales over $[0,1]$, except that the condition of preserving all sups is replaced by
the stronger condition of monotonicity plus continuity (in the usual
sense)~\cite{HajekBook}.

\section{Logics for the Lawvere Quantale}
\label{sec:LLQ}

In this section, we present three propositional logics interpreted over the 
Lawvere quantale which we will collectively refer to as \emph{logics for the Lawvere quantale} (\LLQ).

\textbf{Syntax of logical formulas.}
Formulas are freely generated from a set $\Prop= \{p_1, p_2, \dots \}$ of atomic propositions over logical connectives that can be interpreted in the Lawvere quantale: 
\begin{align*}
 &\bot \mid \top \mid
 \phi \land \psi \mid 
 \phi \lor \psi \mid
  \phi \ot \psi  \mid 
 \phi \lol \psi 
 \tag{quantale connect.}
 \\
 &\un \tag{constant} 
 \\
 &r * \phi \quad \text{(for $r \in \preals$)} \tag{scalar multiplication}
\end{align*}

The first logic, $\LL$, uses only the basic logical connectives that can 
be interpreted in any commutative quantale, viz, the constants bottom ($\bot$) and top ($\top$), binary conjunction ($\land$) and disjunction ($\lor$), tensor ($\ot$), and linear implication ($\lol$). 

The second logic, $\LLun$, additionally allows the use of the constant $\un$. 

The third logic, $\LLs$, extends the syntax further with scalar multiplication by a non-negative real ($r * -$).


It will be useful to define, in all \LLQ, the following derived connectives:
\begin{align*}
\lnot\phi &\coloneq \phi\lol\bot \,, 
\tag{Negation} \\
\phi \lollol \psi &\coloneq (\phi \lol \psi) \land (\psi \lol \phi) \,.
\tag{Double implication}
\end{align*}
Moreover, for any $n \in \naturals$, the derived connective $n \phi$ is inductively defined as follows
\begin{align*}
 0 \phi \coloneq \top && \text{and} && (n+1) \phi \coloneq \phi \ot n \phi \,.
\end{align*}
Thus $\LLun$ effectively has all natural numbers as constants via the formulas
$r*\un$. (In fact $\LLun$ is equivalent to the logic obtained from $\LL$ by adding all
rationals as constants.)  Similarly, for any $r \in \preals$, we write simply $r$ to
denote the formula $r * \un$. 

\begin{notn}
  To simplify the presentation, we assume an operator precedence rule so that $*$ binds
  most strongly, followed by $\ot$, next are $\land$ and $\lor$, and the weakest are
  $\lol$, $\lollol$ and $\lnot$.  Thus, the formula $r*\phi\ot\psi\land s*\psi\lol\theta$
  is interpreted as $(((r*\phi)\ot\psi)\land (s*\psi))\lol\theta$.
\end{notn}

\textbf{Semantics of logical formulas.}
The models of \LLQ\ are maps $\M \colon \Prop \to \extreals$ interpreting the 
propositional symbols in the Lawvere quantale, which can be extended uniquely 
to formulas by setting 
\begin{align*}
\begin{aligned}
  \M(\bot) 		&\coloneq \infty \,, \\
  \M(\top) 		&\coloneq 0 \,, \\
  \M(\un)		&\coloneq 1 \,, \\
  \M(r * \phi)	&\coloneq r \M(\phi) \,,
\end{aligned}
&&
\begin{aligned}
  \M(\phi \land \psi)	&\coloneq \max \{ \M(\psi), \M(\phi) \} \,, \\
  \M(\phi \lor \psi)	&\coloneq \min \{ \M(\psi), \M(\phi) \}  \,, \\
  \M(\phi \ot \psi)		&\coloneq \M(\phi) + \M(\psi) \,, \\
  \M(\phi \lol \psi)		&\coloneq \M(\psi) \dotdiv \M(\phi) \,,
\end{aligned}
\end{align*}
with derived connectives $\lnot$ and $\lollol$ interpreted as
\begin{align*}
  \M(\lnot\phi) = \infty \dotdiv \M(\phi) \,,
  &&
  \M(\phi \lollol \psi) = |\M(\psi) - \M(\phi)| \,.
\end{align*}

\section{Natural Deduction Systems}
\label{sec:proofSys}

We present natural deduction systems for the thee logics 
$\LL \subseteq \LLun \subseteq \LLs$. As each logic is intended to be a conservative
extension of its sub-logics, we present their deduction systems incrementally. 

Let $\mathcal{L} \in \{ \LL, \LLun, \LLs\}$.
A \emph{judgement} in $\mathcal{L}$ is a syntactic construct of the form
\begin{equation*}
 \phi_1, \dots, \phi_n \vdash \psi \,,
 \tag{Judgement}
\end{equation*}
where the $\phi_i$ and $\psi$ are logical formulas of $\mathcal{L}$, called respectively the \emph{antecedents} 
and the \emph{consequent} of the judgement. Note that the antecedents $\Gamma = (\phi_1, \dots,
\phi_n)$ of a judgement form a finite ordered list, possibly, with repetitions.  
As is customary, for $\Gamma$ and $\Delta$ lists of formulas, their comma-separated
juxtaposition $\Gamma, \Delta$ denotes concatenation; and $\vdash \psi$ is the notation
for a judgement with empty list of antecedents. 

A judgement $\gamma = (\Gamma \vdash \psi)$ in $\mathcal{L}$ \emph{is satisfied by} a
model $\M$, in symbols $\M \models_{\mathcal{L}} \gamma$, whenever 
\begin{equation*}
\textstyle
 \sum_{\phi \in \Gamma} \M(\phi) \geq \M(\psi) \,.
 \tag{Semantics of judgements}
\end{equation*}
When the logic $\mathcal{L}$ is clear from the context or the satisfiability holds in all
\LLQ, we simply write $m \models \gamma$. 

A judgement is \emph{satisfiable} if it is satisfied by a model; \emph{unsatisfiable} if
it is not satisfiable; and a \emph{tautology} if it is satisfied by all models. 
Note that, for any model $\M$ in \LLQ
\begin{gather*}
\begin{aligned}
\M \models (\vdash \phi) &&\text{iff} && \M(\phi) &= 0 \\
\M \models (\vdash \lnot \phi) &&\text{iff} && \M(\phi) &= \infty \quad \text{(\ie, $\phi$ is infinite)} \\
\M \models (\vdash \lnot \lnot \phi) &&\text{iff} && \M(\phi) &< \infty 
					\quad \text{(\ie, $\phi$ is finite)} \\
\M \models (\phi \vdash \psi) &&\text{iff}&& \M(\phi) &\geq \M(\psi) \,.
\end{aligned}
\end{gather*}
In particular, $\vdash \phi \lol \phi$, $\vdash \top$, and 
$\vdash \lnot \bot$ are examples of tautologies, while $\vdash \phi \lollol (\lnot\lnot\phi)$
is not. Moreover, by using negation we can express whether the interpretation 
of a formula is either finite or infinite.

An \textit{inference (rule)} is a syntactic construct of the form $\infrule{S}{~\gamma~}$,
for $S$ a set of judgements and $\gamma$ a judgement. 
The judgements in $S$ are the \emph{hypotheses of the inference} and 
$\gamma$ is the \emph{conclusion of the inference}. When $S=\{\gamma'\}$ is a singleton, we write 
\begin{equation*}
\begin{aligned}
\doubleinfrule{\gamma'}{~\gamma~}
\end{aligned}
\quad \text{to denote that both} \quad
\begin{aligned}
\infrule{\gamma'}{~\gamma~}
\end{aligned}
\;
\text{and }
\;
\begin{aligned}
\infrule{\gamma}{~\gamma'~}
\end{aligned}\;
\text{hold.}
\end{equation*}

A judgement $\gamma$ is a \emph{semantic consequence} of a set $S$ of judgements, in
symbols $S \models \gamma$, if every model that satisfies all the judgements in $S$
satisfies also $\gamma$. Thus, $\emptyset \models \gamma$ (or more simply
$\models \gamma$) means that $\gamma$ is a tautology. For a model $\M$, we will also use
the notation $S \models_\M \gamma$, to mean that, whenever $m$ satisfies all the
judgements of $S$, then it satisfies also $\gamma$.

\paragraph*{The natural deduction system of $\LL$}
It consists of the inference rules in Figure~\ref{fig:LL}. Figure~\ref{tab:structRules}
contains the basic rules of logical deduction (\textsc{id}) and (\textsc{cut}), and the
structural rules of weakening (\textsc{weak}) and permutation (\textsc{perm}). Note that,
there is no cancellation rule. Figure~\ref{tab:latticeRules} provides the rules for the
lattice operations of the Lawvere quantale\footnote{Recall that the order on $\extreals $
  is reversed!}. Figure~\ref{tab:lawvereRules} collects the rules that are specific to the
Lawvere quantale. (\textsc{wem}) is the weak excluded middle; (\textsc{tot}) states that
the quantale is totally ordered; the other rules explain the actions of $\ot$ and its
adjoint in the Lawvere quantale. 

\paragraph*{\textbf{The natural deduction system of $\LLun$}}
It includes all the rules in Figure~\ref{fig:LL} and, in addition,
\begin{equation*} 
\infrule[one]{\vdash \un \lor \lnot\un}{ \vdash \bot}
\end{equation*}
expressing that $0 \geq 1$ and $1 \geq \infty$ are inconsistencies. Thus, 
with $\LLun$ we exit the universe of classical logic, as $\un$ cannot be 
provably equivalent either to $\top$, or to $\bot$.

\paragraph*{\textbf{The natural deduction system of $\LLs$}}
It extends the deduction system of $\LLun$ with the rules for scalar multiplication in
Figure~\ref{tab:scalarRules}. In ($S_4$), ${\bowtie}$ can be either of   
$\land,\lor,\otimes, \lol$, meaning that we have one version of ($S_4$) for each of these operators.

\begin{figure}[tp]
\centering
\begin{subfigure}{0.4\textwidth}
\hrule
\begin{gather*}
\infrule[id]{}{\phi \vdash \phi}
\\[2ex]
\infrule[cut]{
	\Gamma \vdash \phi
	&
	\Delta,\phi \vdash \psi
}{ \Gamma,\Delta \vdash \psi }
\\[2ex]
\infrule[weak]{
	\Gamma \vdash \phi
}{\Gamma,\psi \vdash \phi}
\\[2ex]
\infrule[perm]{
	\Gamma, \phi, \psi, \Delta \vdash \theta
}{\Gamma, \psi, \phi, \Delta \vdash \theta}
\end{gather*}
\hrule\vspace{1ex}
	\caption{Logical deduction and Structural rules}
	\label{tab:structRules}
\end{subfigure}
\quad
%
\begin{subfigure}{0.4\textwidth}
\hrule
\begin{align*}
\infrule[top]{}{\Gamma \vdash \top}
&&&
\infrule[bot]{}{\bot \vdash \phi}
\\[2ex]
\infrule[$\land_1$]{
	\Gamma, \phi \vdash \theta
}{\Gamma, \phi \land \psi \vdash \theta}
&&&
\infrule[$\lor_1$]{
	\Gamma, \phi \vdash \theta
	&
	\Gamma, \psi \vdash \theta
}{\Gamma, \phi \lor \psi \vdash \theta}
\\[2ex]
\infrule[$\land_2$]{
	\Gamma \vdash \phi
	&
	\Gamma \vdash \psi
}{\Gamma \vdash \phi \land \psi}
&&&
\infrule[$\lor_2$]{
	\Gamma \vdash \phi
}{\Gamma \vdash \phi \lor \psi}
\\[2ex]
\infrule[$\land_3$]{
	\Gamma \vdash \phi \land \psi
}{\Gamma \vdash \psi}
&&&
\infrule[$\lor_3$]{
	\Gamma, \phi \lor \psi \vdash \theta
}{\Gamma, \psi \vdash \theta}
\end{align*}
\hrule\vspace{1ex}
	\caption{Lattice rules}
	\label{tab:latticeRules}
\end{subfigure}

\begin{subfigure}{0.858\textwidth}
\hrule
\begin{align*}
\infrule[wem]{}{\vdash (\lnot\phi) \lor (\lnot\lnot\phi)}
&&&
\infrule[tot]{}{\vdash (\phi \lol \psi) \lor (\psi \lol \phi)}
\\[2ex]
\doubleinfrule[$\ot_1$]{
	\hfill\Gamma, \phi, \psi \vdash \theta
}{\Gamma, \phi \ot \psi \vdash \theta}
&&&
\infrule[$\lol_1$]{
	\Gamma, \phi \lol \theta \vdash \psi
	&
	\theta \vdash \phi
}{\Gamma, \theta \vdash \phi \ot \psi}
\\[2ex]
\doubleinfrule[$\ot_2$]{
	\Gamma, \phi \ot \psi \vdash \theta
}{\Gamma, \phi \vdash \psi \lol \theta}
&&&
\infrule[$\lol_2$]{
	\Gamma, \theta \vdash \phi \ot \psi
	&
	\vdash \lnot\lnot\phi
}{\Gamma, \phi \lol \theta \vdash \psi}
\\[2ex]
\infrule[$\ot_3$]{
	\phi \ot \phi \vdash \psi \ot \psi
}{\phi \vdash \psi}
&&&
\infrule[$\lol_3$]{
	\Gamma, \theta \vdash \phi \ot \psi
	&
	\vdash \lnot\lnot\theta
}{\Gamma, \phi \lol \theta \vdash \psi}
\end{align*}
\hrule\vspace{1ex}
	\caption{Lawvere quantale rules}
	\label{tab:lawvereRules}
\end{subfigure}
\caption{The natural deduction system of $\LL$}
\label{fig:LL}
\end{figure}

\begin{figure}[tp]
\centering
\begin{subfigure}{0.858\textwidth}
\hrule
\begin{gather*}
\begin{aligned}
\doubleinfrule[$S_1$]{\phi\vdash\psi & (r > 0)}{r*\phi\vdash r*\psi}
&\quad&\quad&
\infrule[$S_2$]{}{\vdash r*(s*\phi) \lollol (rs)*\phi} 
\\[1ex]
\infrule[$S_3$]{}{\vdash \phi \lollol 1*\phi}
&\quad&\quad&
\infrule[$S_4$]{{\bowtie} \in \{\land,\lor,\otimes, \lol\}}{\vdash r*(\phi\bowtie\psi) \lollol r*\phi\bowtie r*\psi } 
\\[1ex]
\infrule[$S_5$]{}{\vdash 0*\phi }
&\quad&\quad&
\infrule[$S_6$]{}{\vdash (r+s)*\phi \lollol r*\phi\ot s*\phi }
\\[1ex]
\infrule[$S_7$]{}{r*\bot \vdash \bot }
&\quad&\quad&
\infrule[$S_8$]{}{\vdash (r\dotdiv s)*\phi \lollol (r*\phi\lol s*\phi) }
\end{aligned}
\\[1ex]
\infrule[$S_9$]{}{\vdash r*\phi\land s*\phi \lollol \max\{r,s\}*\phi  }
\\[1ex]
\infrule[$S_{10}$]{}{\vdash r*\phi \lor s*\phi \lollol \min\{r,s\}*\phi  }
\end{gather*}
\hrule
\end{subfigure}
	\caption{The natural deduction system of $\LLs$ (scalar product rules)}
	\label{tab:scalarRules}
\end{figure}

\begin{definition}
Let $S$ be a set of judgements in $\mathcal{L} \in \{ \LL, \LLun, \LLs\}$. 
We say that a judgement $\gamma$ is \emph{provable from} 
(or \emph{deducible from}) $S$ in $\mathcal{L}$
(in symbols $S \Vdash_{\mathcal{L}} \gamma$), if there exists a 
sequence $\gamma_1, \dots, \gamma_n$ of judgements 
ending in $\gamma$ whose members are either 
an axiom of $\mathcal{L}$, or a member of $S$, or it follows from some preceding members of the sequence by using the inference rules in $\mathcal{L}$.
A sequence $\gamma_1, \dots, \gamma_n$ as above is called a \emph{proof}.

A judgement $\gamma$ is a \emph{theorem} of $\mathcal{L}$ if it is provable in
$\mathcal{L}$ from the empty set (denoted $\emptyset \Vdash_{\mathcal{L}} \gamma$, or
$\Vdash_{\mathcal{L}} \gamma$). 
\end{definition}

%
\begin{theorem}[Soundness of \LLQ]\label{soundness}
Let $\mathcal{L} \in \{ \LL, \LLun, \LLs\}$. If a judgement $\gamma$ is 
provable from $S$ in $\mathcal{L}$, then $\gamma$ is a semantic consequence 
of $S$ in $\mathcal{L}$
(in symbols, $S \Vdash_{\mathcal{L}} \gamma$ implies $S \models_{\mathcal{L}} \gamma$).
\end{theorem}

\begin{notn}
In \LLQ\ we can prove that $\ot$ is associative, commutative, and with $\top$ 
as null element. Thus, hereafter we will write $\phi_1\ot \dots \ot\phi_n$,
or sometimes $\bigotimes_{i\leq n} \phi_i$, 
without involving unnecessary parenthesis, as the notion is unambiguous.
This includes the case $n = 0$, where we interpret $\bigotimes_{i\leq 0} \phi_i = \top$.
\end{notn}

Note that, any judgement 
$\phi_1, \dots, \phi_n \vdash \psi$ is provably equivalent to 
$\vdash (\phi_1 \ot \dots \ot \phi_n) \lol \psi$. Thus, without loss of generality,
we may assume judgements are always of the form $\vdash \phi$, for some $\phi$.

\subsection{Deduction Theorems}
Classical and intuitionistic logics both enjoy the \emph{deduction theorem}: 
if $\phi, \psi$ are formulas and $S$ a set of judgements, $\vdash \psi$ is provable from 
$S \cup \{\vdash \phi\}$ iff $\vdash \phi \to \psi$ is provable from $S$.
In \LLQ\ ---similarly to other substructural logics without a cancellation rule--- 
(the left-to-right implication) does not hold.
\begin{fact}[Failure of the deduction theorem] 
Consider the formulas $\phi \coloneq \eta \land((\eta\ot\rho)\lol\theta)$ 
and $\psi \coloneq \rho \lol \theta$. Then $\vdash \psi$ is provable from 
$\{\vdash \phi\}$, as follows.
\begin{equation*}
\infrule[$\ot_2$]{
	\infrule[cut]{ 
		\infrule[$\land_3$]{\vdash \phi}{\vdash \eta}
		&
		\infrule[$\land_3$ ; $\ot_2$ ; $\ot_1$]{\vdash \phi
			}{\eta, \rho \vdash \theta}
	}{ \rho \vdash \theta}
}{\vdash \psi}
\end{equation*}
But $\vdash \phi \lol \psi$ is not provable, because otherwise, from the soundness, it
should be a tautology and it is not: consider the model $\M$ such that
$\M(\eta) = \frac14$, $\M(\rho) =0$ and $\M(\theta) = 1$.  
\end{fact}

Similarly to other substructural logics like linear or {\L}ukasiewicz logics, for $\LL$
a weaker form of the deduction theorem holds: $\vdash \psi$ is provable from $S \cup \{\vdash \phi\}$ in $\LL$ iff $\vdash n\phi \lol \psi$ is provable from $S$ in $\LL$ for some $n \in \naturals$. This is a consequence of the equivalence of $\LL$ with product logic, for which such a deduction theorem holds~\cite[pg.196]{H96}.
However, this weaker version does not hold in $\LLun$ and $\LLs$.
\begin{fact}[Failure of the weak deduction theorem] \label{fact:wdt}
Consider the formulas $\phi \coloneq \un \lor \lnot\un$ 
and $\psi \coloneq \bot$. Then $\vdash \psi$ is provable from $\{\vdash \phi\}$
using \textsc{(one)}.
But $\vdash n\phi \lol \psi$ is not provable for any $n \in \naturals$, 
since otherwise, using the soundness, there should exist
an $n$ such that $\vdash n\phi \lol \psi$ is a tautology.
However, no model satisfies this judgement. Indeed, for any model $\M$, 
\begin{equation*}
\M(n\phi \lol \psi)
	= \M(\bot) \dotdiv n \min\{ \M(\un), \M(\bot)\} \\
	= \infty \dotdiv n = \infty \,.
\end{equation*}
\end{fact}
Fact~\ref{fact:wdt} shows that $\LLun$ is not equivalent to any of the
  extensions of product logic with truth-constants proposed in~\cite{EG00}, for which such
  a weak deduction theorem does hold.

\subsection{Totality Lemma}
\label{sec:totalityLemma}

We call pairs of judgements of the following form
\begin{align*}
( \vdash\phi\lol\psi \; , \; \vdash\psi\lol\phi )
&& \text{ or } &&
(\vdash\lnot\phi \; , \; \vdash\lnot\lnot\phi )
\end{align*}
\emph{supplementary judgements}. These judgements, when used as a pair as above,
explore different alternatives for the interpretations of \LLQ-formulas.
The first pair of judgements explores different ordering alternatives; 
the second is choosing 
either a finite or infinite interpretation for a formula.

Supplementary judgements play a special r\^ole in reasoning, 
as clearly stated in the following lemma.
\begin{lemma}[Totality Lemma]
\label{totalityl1}
The following statements are provable in all \LLQ.
\begin{enumerate}[itemsep={1ex}, topsep={1ex}]

\item 
$\text{If }
\Big[
\dfrac{\vdash\rho~~~~\phi\vdash\psi}{\vdash\theta}
\text{ and }
\dfrac{\vdash\rho~~~~\psi\vdash\phi}{\vdash\theta}
\Big]
\text{ then }
\dfrac{\vdash\rho}{\vdash\theta}$.


\item 
$\text{If }
\Big[
\dfrac{\vdash\rho~~~~\vdash\lnot\phi}{\vdash\theta}
\text{ and }
\dfrac{\vdash\rho~~~~\vdash\lnot\lnot\phi}{\vdash\theta}
\Big]
\text{ then }
\dfrac{\vdash\rho}{\vdash\theta}$.
\end{enumerate}	
\end{lemma}

The totality lemma, along with others, allows us to prove the distributivity of the tensor
product with respect to $\land$ and $\lor$, respectively. 

\section{Theories and Models}
\label{sec:theories}

In what follows, we use $\mathcal{L}$ to range over $\{\LL,\LLun,\LLs\}$, 
as the following definitions are uniform for all \LLQ.

A \emph{theory} $\T$ in $\mathcal{L}$ is a set of judgements that is
deductively closed (in symbols, $\T \Vdash_{\mathcal{L}} \gamma$ implies 
$\gamma \in \T$). An \emph{axiomatic theory} in $\mathcal{L}$ is a theory 
for which there exists a set of judgements, called \emph{axioms}, such that 
all the judgements in the theory can be proven in $\mathcal{L}$ from the axioms;
it is \emph{finitely axiomatized} if it admits a finite set of axioms.

If $\T$ and $\T'$ are two theories in $\mathcal{L}$ such that $\T\subseteq \T'$, 
we say that $\T'$ is a \emph{extension} of $\T$; it is a \emph{proper extension} 
if $\T\subsetneq \T'$.

A theory $\T$ in $\mathcal{L}$ is \emph{disjunctive}, if for any formulas
$\phi,\psi\in\mathcal{L}$, $\vdash\phi\lor\psi\in\T$ implies that either $\vdash\phi\in\T$
or $\vdash\psi\in\T$. It is immediate that if $\T$ is a disjunctive theory,
because of (\textsc{tot}) and (\textsc{wem}), we have that for any set of supplementary
judgements in $\mathcal{L}$, at least one of the judgements belongs to $\T$.

A theory in $\mathcal{L}$ is \emph{inconsistent} if it contains $\top\vdash\bot$, 
otherwise it is \emph{consistent}; it is \emph{maximal consistent} if it is consistent 
and all its proper extensions are inconsistent.

A \emph{model of a theory $\T$} is a model $\M$ that satisfies all the judgements of the
theory. If the theory is axiomatized, $\M$ is a model for all the axioms iff it is a model
of the theory.

Note that an assignment of values to all the propositional atoms defines a unique model
since the values of all formulas are given inductively and are determined by the values of
the atomic propositions.

\begin{lemma}\label{max}
In all \LLQ\ the following statements are true.
\begin{enumerate}
\item If a theory has a model, then it is consistent.
\item Any model satisfies a disjunctive consistent theory.
\end{enumerate}
\end{lemma}

In the case of $\LLs$ we can identify a special class of 
disjunctive consistent theories. 
\begin{definition}
A \emph{diagrammatic theory} is a consistent theory $\T$ of $\LLs$ 
such that for any $p\in\Prop$,
\begin{itemize}
	\item either $p\vdash\bot\in\T$,
	\item or there exists $\e\in\preals$ such that $\e\vdash p\in\T$ and $p\vdash \e\in\T$.
\end{itemize}
\end{definition}

It is not difficult to observe that in a diagrammatic theory, for any $\phi\in\LLs$,
either $\phi\vdash\bot\in\T$, or there exists $\e\in\preals$ such that 
$\e\vdash \phi$,\;  $\phi\vdash \e\in\T$. 

\begin{lemma}\label{characteristictheory}
In $\LLs$ we have that
\begin{enumerate}
  \item Every diagrammatic theory has a unique model. 
  \item Every model satisfies a unique diagrammatic theory.
  \item A theory is diagrammatic iff it is maximal consistent.
\item Every disjunctive consistent theory has a unique diagrammatic extension; and a unique model.
	
\end{enumerate}
\end{lemma}

\section{Normal Forms}
\label{sec:normalform}
In this section, we prove that any finitely axiomatized theory can be presented 
in a normal form, where all the axioms have a specific syntactic format.

There are some important classes of judgements that play a crucial role in our development:
\begin{align*}
\tag{tautological}
&\bot\vdash\phi \;\mid\; \phi\vdash\top \\
\tag{inconsistent}
&\top\vdash\bot \;\mid\; \top\vdash\un \;\mid\; \un\vdash\bot \\
\tag{assertive} 
&\underbrace{\top\vdash p \;\mid\;  p\vdash\bot}_{\text{alethic}}
\;\mid\; 
\underbrace{\vdash\lnot\lnot p}_{\text{finitist}} 
\\[0.5ex]
&\textstyle
(\bigotimes_{i \leq n} r_i*p_i ) \ot r*\un \vdash (\bigotimes_{j \leq m} s_i*q_i ) \ot s*\un
\tag{affine}
\end{align*}
where $p, p_i, q_j \in \Prop$ are atomic propositions. 
In the case of $\LL$ and $\LLun$ the coefficient in an affine judgement are positive integers, and for $\LL$ the term involving $\un$ is not present. 

\begin{definition}[Normal form]
A judgement is in \emph{normal form} if it is either tautological,  inconsistent, assertive (finitist or alethic), or affine.
\end{definition}


\begin{notn}
Since in $\LLs$ $\ot$ commutes with all the other logical connectives, we assume hereafter that in all the formulas the scalar products 
guard the atomic propositions or the constants, and no other scalar products appear in a formula.
\end{notn}

\begin{definition}\label{decisiontree}
A theory in \LLQ\ is \emph{normal}, or it has a \emph{normal axiomatization}, 
if it admits a finite axiomatization such that
\begin{itemize}
  \item every axiom is in normal form;
  \item no atomic proposition that occurs in an alethic axiom appears 
  in any other axiom;
  \item there is an assertive judgement for each atomic proposition 
  that appears in the axioms.
\end{itemize}
\end{definition}

For \LLQ\ it is not possible, in general, to convert a judgement into a model-theoretic equivalent judgement in normal form. It is however possible to associate with any judgement $\gamma$, a finite (possibly empty) set of normal theories 
$\T_1, \dots, \T_n$, such that:
\begin{align*}
\M \models \gamma 
&&\text{iff}&& 
\text{for some $i\leq n$, } \, \M \models \T_i \, 
 \,. 
\end{align*}
In such a case, we call the set $\{\T_1, \dots, \T_n \}$ of theories
 a \emph{normal representation of the judgement} $\gamma$.
Similarly, a \emph{normal representation of a finite set of judgements $V$} (or of a finitely axiomatized theory) is a finite (possibly empty) set of normal theories 
$\T_1, \dots, \T_n$, such that, 
$\M$ is a model of $V$ iff it is a model for at least one of the theories 
$\T_1, \dots, \T_n$.

\subsection{Normalization Algorithm}
There exists a simple algorithm that allows us to compute, for any given 
finite set of judgements $V$, its normal representation $\N(V)$.
The details of this algorithm are in the appendix. In what follows, we sketch how 
this algorithm works on a couple of examples that illustrate the main 
subtleties of the algorithm.

Suppose we have a finite set $V$ of judgements. If $V$ is not already in 
normal form, we can use the theorems of \LLQ\ to simplify the judgements and 
eventually convert them to a normal form.
In doing this, in some cases, we will have to use pairs $(\gamma_1, \gamma_2)$ 
of supplementary judgements (see Section~\ref{sec:totalityLemma}) that will be treated as new axioms.
This is done when the conversion cannot
progress without extra assumptions. One can think of it as ``a proof
by cases'' resulting in two new separate set of judgements $V_1$ and $V_2$,
each containing one of the supplementary judgements.
The invariant preserved in each reduction step is that, for $i \in \{1,2\}$
\begin{itemize}
  \item all the judgements in $V_i$ are provable 
  from $V \cup \{ \gamma_i \}$;
  \item all the judgements in $V$ are provable from $V_i$.
\end{itemize}

\begin{example} \label{ex:rule3}
Let $\gamma = \theta \vdash (\phi \lor \psi) \ot \rho$ be the judgement we
would like to reduce to normal form. The disjunction occurring in $\gamma$
is problematic as it prevents $\gamma$ to be provably equivalent to another
(single) judgement in normal form. However, by using the supplementary 
hypotheses $\psi \vdash \phi$ and $\phi \vdash \psi$, we can split the 
reduction by cases and obtain $V_1 = \{ \phi\vdash\psi, \theta \vdash \psi\ot\rho\}$
and $V_2 = \{ \psi\vdash\phi, \theta\vdash \phi\ot\rho\}$, two sets of judgements
on which each element is (at least) one step closer to be in normal form 
(Fig.~\ref{fig:rule3}). 
Note that the invariant described above is preserved.
\end{example}

\begin{figure}[tb]

\centering
\begin{subfigure}{0.3\textwidth}
\begin{tikzpicture}
	[
	grow                    = right,
	sibling distance        = 3em,
	level distance          = 6.3em,
	edge from parent/.style = {draw, -latex},
	every node/.style       = {font=\scriptsize},
	sloped
	]
	\node[root] {$\theta\vdash (\phi\lor\psi)\ot\rho$}
	child { node [env, label=right:$V_2$] 
		{$\phi\vdash\psi$\\$\theta\vdash \psi\ot\rho$}
		edge from parent node [below] {} }
	child { node [env,label=right:$V_1$] 
		{$\psi\vdash\phi$\\$\theta\vdash \phi\ot\rho$}
		edge from parent node [above] {} };
\end{tikzpicture}
\caption{}
\label{fig:rule3}
\end{subfigure}
\begin{subfigure}{0.5\textwidth}
\begin{tikzpicture}
	[
	grow                    = right,
	sibling distance        = 5em,
	level distance          = 6.3em,
	edge from parent/.style = {draw, -latex},
	every node/.style       = {font=\footnotesize},
	sloped
	]
\node [root] {$\theta\vdash (\phi\lol\psi)\ot\rho$}
[child anchor=west]
child { node [env, label=right:$W_1$] {$\phi\vdash\psi$\\$\theta\vdash\rho$}
	edge from parent node [below] {} }
child { node [env] {$\psi\vdash\phi$\\$\theta\vdash (\phi\lol\psi)\ot\rho$}
	child { node [env] {$\vdash\lnot\psi$\\$\psi\vdash\phi$\\
			  	    $\theta\vdash (\phi\lol\psi)\ot\rho$}
		child { node [env, label=above:$W_4$] {$\vdash\lnot\phi$\\
		$\vdash\lnot\psi$\\
					    $\theta\vdash\rho$}
			edge from parent node [below] {} }
		child { node [env, label=above:$W_3$] {$\vdash\lnot\lnot\phi$\\
					    $\vdash\lnot\psi$\\$\vdash\lnot\theta$}
			edge from parent node [above] {}
		}
		edge from parent node [below] {} }
	child { node [env, label=right:$W_2$] {$\vdash\lnot\lnot\psi$\\$\psi\vdash\phi$\\
				    $\theta\ot \phi\vdash \psi\ot\rho$}
		edge from parent node [above, align=center]
		{}}
	edge from parent node [above] {} };
\end{tikzpicture}
\caption{}
\label{fig:rule5}
\end{subfigure}

\caption{Conversion into normal representation: (a) Rule 3 and (b) Rule 5 of the
  normalization algorithm}
\label{fig:conversionNormal}
\end{figure}

\begin{example}\label{ex:rule5}
Let $\gamma = \theta\vdash (\phi\lol\psi)\ot\rho$ be the judgement to be converted 
into normal form. In this specific case, the problematic connective is $\lol$. 
By adding appropriate pairs of supplementary judgements, in sequence, 
we split the reduction in four cases and obtain $W_1, \dots, W_4$ as new sets of judgements (Fig~\ref{fig:rule5}).

Of interest in this particular case, is that in order to guarantee that the new sets 
of judgements have strictly reduced complexity ---interpreted as number of 
sub-formulas not in normal form--- we need to take several reduction steps.
\end{example}

Starting from a finite set of judgements $V$, the normalization algorithm works
essentiality by repeatedly applying conversion rules to the judgements that are not in
normal form by inspecting the structure of the formulas in the judgements. Note that,
Examples~\ref{ex:rule3} and \ref{ex:rule5} describe actual conversion rules in the
algorithm (for the other rules see the full version of the paper on ArXiV).  As each
conversion rule guarantees that the number of sub-formulas not in normal form is strictly
reduced, the algorithms eventually terminates.

The output $\N(V)$ of the algorithm is a set of theories (technically, only their
axioms). The next theorem states the correctness of this conversion. 
\begin{theorem}[Normal representation]\label{normalrep}
Given a finite set $V$ of judgements in $\mathcal{L} \in \{\LL,\LLun,\LLs\}$, the set of the theories axiomatized by the elements in $\N(V)$ is a normal representation of the theory axiomatized by $V$.
Consequently, any model of $V$ is a model for at least one of the elements in $\N(V)$; and any model of an element in $\N(V)$ is a model of $V$.
\end{theorem}
 
 The normalization algorithm also allows us to prove the decidability of satisfiability in \LLQ.
 
\begin{theorem}[Decidability of satisfiability in \LLQ]
Given a finite set $V$ of judgements in $\mathcal{L}\in\{\LL,\LLun,\LLs\}$, $V$ is satisfiable iff there exists $S\in\N(V)$ s.t. $\vdash\bot\not\in S$. Consequently, the satisfiability of judgements in \LLQ\ is decidable.
\end{theorem}

\section{Completeness and Incompleteness}
\label{sec:completeness}

In this section, we demonstrate first that all the logics for the Lawvere quantale are
incomplete in general, even for theories over finitely many propositional
symbols. Secondly, we prove that all \LLQ\ are complete if consider finitely-axiomatised
theories only. Finally, we prove an approximate form of strong completeness over a well
behaved class of theories, not necessarily finitely-axiomatizable.


\begin{theorem}[Incompleteness]
\label{th:incompleteness}
LLQ are incomplete: for any $\mathcal{L}\in\{\LL,\LLun\LLs\}$, there exist theories $\T$
and judgements $\gamma$ in $\mathbb L$ so that all the models of $\T$ are models of
$\gamma$ but $\gamma$ is not provable from $\T$ in $\mathcal{L}$. Moreover, the result
is independent of the particular proof systems that are chosen for LLQ, in the sense
that any finite set of finitary proof rules that can be proposed (as an alternative to the
rules presented in this paper) still produces an incomplete theory for each $\mathcal{L}$.  
\end{theorem}  
  
\begin{proof}
Consider $\mathcal{L}\in\{\LL,\LLun\LLs\}$ with their proof systems presented in
Section~\ref{sec:proofSys}, or any alternative finite set of  \emph{finitary} rules that
can describe $\mathcal{L}$. 
Let $p,q\in\Prop$ be two atomic propositions and $\T$ a theory in 
$\mathcal{L}$ axiomatized by all the judgements of the form
$
(n+1)p\vdash n q \quad \text{for all }n \in \naturals \,.
$

Note that in all the models $\M$ of $\T$ we must have $\M(p) \geq \M(q)$, 
hence all the models of $\T$ are also models of $p\vdash q$.
%
Assume there exists a finite proof of $p\vdash q$ in 
$\mathcal{L}$ from the set $\{(n+1)p\vdash q\mid n\geq 0\}$ of axioms of $\T$. 
Since this proof is finite and uses a finite set of finitary rules, there must exist 
$k\geq 0$ so that the only judgements used in the proof of $p\vdash q$ are 
from the set 
$V=\{(n+1)p\vdash q\mid 0\leq n\leq k\}$. If that is the case, then any model of $V$ is a model 
for $p\vdash q$ (from soundness).
But this is false: consider the model $\M$ such that 
$\M(p) = \frac{k}{k+1}$ and $\M(q) =1$. This is a model of $V$, but not a 
model of $p\vdash q$.
\end{proof}

A consequence of Theorem~\ref{th:incompleteness} is that not all consistent theories have
models. For instance in $\LLun$, the theory axiomatized by the set
$\{p\vdash n\mid n\in\mathbb N\}\cup\{\vdash\lnot\lnot p\}$ of axioms for $p\in\Prop$, is
consistent: any proof will use a finite subset of axioms from the first set and possibly
$\vdash\lnot\lnot p$, and these are not sufficient to prove $\vdash\bot$. However, this
theory has no model, because in any model $\M$ the axioms in the first set guarantee that
$\M(p) \geq n$, for all $n\in\naturals$, while the last axiom require that $\M(p)$ is
finite.

%
\begin{theorem}[Completeness for finitely-axiomatized theories]
\label{completeness}
Let $\mathcal{L}\in\{\LL,\LLun\LLs\}$ and $\T$ a finitely-axiomatized theory in
$\mathcal{L}$. If a judgement $\gamma$ is a semantic consequence of $\T$ in $\mathcal{L}$,
then $\gamma$ is provable from $\T$ in $\mathcal{L}$ 
(in symbols, $\T \models_{\mathcal{L}} \gamma$ implies $\T \Vdash_{\mathcal{L}} \gamma$). 
\end{theorem}

\begin{proof} 
The proof is similar for all \LLQ, by adapting only the arguments to the appropriate
context. Hereafter, we sketch the proof for $\LLs$, as it is the most complex of all. We
only sketch this for the case both $\gamma$ and $\T$ are in normal form. The other cases
can all be reduced to this one.  
	
Assume that $\gamma$ is a normal judgement and $\T$ a normal theory. The cases when
$\gamma$ is an alethic or a finitist judgement are relatively simple. We detail here the
case when $\gamma$ is affine: 
\begin{equation*}
(\bigotimes_{i \leq n} r_i*p_i ) \ot r*\un \vdash (\bigotimes_{j \leq m} s_j*q_j ) \ot s*\un
\end{equation*} 
with $m, n$ possibly $0$. We need to show that the judgement above is provable from the axioms of $\T$. 
We have a couple of cases regarding the possibility that some of atomic propositions
$p_i$, $q_j$ appear in alethic axioms in $\T$, but we can easily deal with these
propositions and eventually reduce the problem to the case when none of the atomic
propositions $p_i$, $q_j$ appear in the alethic axioms of $\T$.  
	
	Consequently, all these atomic propositions appear in finitist axioms of $\T$. 	
	Using commutativity and associativity of $\ot$, we will reorganise both our judgement and the non-assertive axioms of $\T$ so that we put together different copies of the same atomic proposition in a tensorial product and use the facts that $0p=\top$ and $r*\top\ot\phi\dashv\vdash\phi$. So, without losing generality, we can assume that our judgement $\gamma$ is  
\begin{equation*}
(\bigotimes_{i \leq k} a_i*x_i ) \ot r*\un \vdash (\bigotimes_{i \leq k} b_i*x_i ) \ot s*\un
\end{equation*} 
and the non-assertive axioms of $\T$ are 
\begin{equation*}
\left\{
\begin{aligned} 
(\bigotimes_{i \leq k} a^1_i*x_i ) \ot r^1*\un 
&\vdash (\bigotimes_{i \leq k} b^1_i*x_i ) \ot s^1*\un
\\
&\dots
\\
(\bigotimes_{i \leq k} a^l_i*x_i ) \ot r^l*\un 
&\vdash (\bigotimes_{i \leq k} b^l_i*x_i ) \ot s^l*\un
\end{aligned}
\right.
\end{equation*}	
for some positive reals $a_i, b_i, a_i^j, b_i^j, r, s, r^j, s^j$ and atomic propositions $x_1, \dots, x_k$.
Consider the matrices $A \in \reals^{l \times k}$, 
$C \in \reals^{k\times 1}$ and vector $\beta \in \reals^{k}$.
\begin{align*}
A=\begin{pmatrix} 
	a^1_1-b^1_1 & \dots & 
	a^1_k-b^1_k \\
	\dots &\dots&\dots
	\\
	a^l_1-b^l_1 &\dots& a^l_k-b^l_k 
\end{pmatrix}
&&
C = (b_1 - a_1, \dots, b_k - a_k) \,
&&
\beta =\begin{pmatrix} r^1-s^1\\ \dots \\r^l-s^l\end{pmatrix}
\end{align*}
and let $\delta = r-s$.
According to our hypothesis, any model of $\T$ is a model of $\gamma$, it follows that
there exists no $x = (x_1,\dots,x_k)\in \mathbb R^k$ such that 
\begin{gather*}
Ax + \beta \geq 0 \,, \quad Cx + \delta >0 \,.
\end{gather*}
By Mozkin transposition theorem~\cite{Motzkin51}%
\footnote{Motzkin~\cite{Motzkin51} contains an unfortunate typo. 
For a correct form of the statement see~\cite{Border2013}.
Here we use an adapted form for affine transformations obtained by using the Fundamental
Theorem of Linear Inequalities (see~\cite[Corollary~7.1h]{Schrijver1998})},
there exist 
$t_0 \in \reals$ and $t=(t_1,\dots,t_l)\in\reals^{l\times 1}$, such that
\begin{gather*}
C x + \delta = t(Ax + \beta) + t_0 \,, \quad t \geq 0 \,, \quad t_0 \geq 0 \,.
\end{gather*}
In $\LLs$ we can repeatedly apply the derived rule following the pattern from $t$, $t_0$
\begin{equation*}
\infrule{\phi_1\vdash\psi_1 & \phi_2\vdash\psi_2
	}{r*\phi_1\ot s*\phi_2\vdash r*\psi_1\ot s*\psi_2}
\end{equation*}
we obtain 
a proof from $\T$ for 
$
\big( (\bigotimes_{i \leq k} a_i*x_i ) \ot r*\un \big) 
\vdash 
\big( (\bigotimes_{i \leq k} b_i*x_i ) \ot s*\un \big)
$.
\end{proof}

A consequence of this completeness result is the following.
\begin{corollary}\label{characteristicmodel}
	For any finitely axiomatized theory of $\LLs$, the set of its diagrammatic extensions
    coincide with the set of the diagrammatic theories of its models. 
\end{corollary}

Incompleteness over general theories (Theorem~\ref{th:incompleteness}) is a common trait of several many-valued
logics~\cite{Marra13,Yaacov13,yaacov10}, especially, if interpreted
over the reals.
A weaker form of completeness result, first proposed by Ben Yaacov~\cite{yaacov10}, 
is \emph{approximate completeness}. Rather than compromise on 
the theories, one asks instead whether a judgement can be ``proven up to 
arbitrary precision''. 

In $\LLs$, approximate completeness can be formally stated as follows: 
whenever all the models of a set of judgements $S$ are also models of 
$\vdash \psi$, the judgement $\epsilon \vdash \psi$ is provable from 
$S$ in $\LLs$, for any $\e > 0$.


It is not difficult to see that the above statement is still too strong
to hold in $\LLs$ for general sets of judgements $S$. 
Actually, already theories using only finitely-many atomic propositions
(one is enough) can falsify the statement.
\begin{fact}[Failure of approximate completeness]
Consider the set $S = \{ p \vdash n \mid n \in \naturals \}$ of judgements in 
$\LLs$, for $p \in \Prop$ a fixed atomic proposition, and take $\psi = \lnot p$.

The only model $\M$ of $S$ is such that $\M(p) = \infty$, because satisfying
all the judgements of the form $p \vdash n$, for $n \in \naturals$, is equivalent
to say that the interpretation of $p$ is $\infty$. Thus, it is also a model
for $\vdash \lnot p$.

Assume that approximate 
completeness holds in $\LLs$ and let $\e < \infty$. 
Then, $\lnot p$ 
is provable from $S$. 
As any proof is finite, there must exists a finite subset $S' \subseteq S$ 
such that $\e \vdash \lnot p$ is provable from $S'$. 

Define $N \coloneq \max \{n \mid p \vdash n \in S' \}$.
Then, $\M'(p) \coloneq N$ is a model for $S'$ and, by Theorem~\ref{soundness}, $\M \models (\e \vdash \lnot p)$.
That is, $\e = \M'(e) \geq \M'(\lnot\phi)$. However, $\M'(\lnot\phi) = \infty$, 
thus, $\M'(\lnot\phi) > \e$  ---contradiction.
\end{fact}

Despite the fact that we cannot hope for a general form of approximate strong
completeness, we can still recover a mildly restricted version of it by focusing on a
suitable well-behaved class of judgements and theories over finitely-many atomic
propositions.

Let $\Prop[n] = \{p_1, \dots, p_n \}$ be a finite set of atomic propositions, 
and denote by $\LLs(n)$ the logic $\LLs$ restricted over $\Prop[n]$. Then, using the
Hurwicz's general form of Farkas' Lemma~\cite{Hurwicz2014}, and following a similar proof
structure as in Theorem \ref{completeness}, we can prove the following approximate
completeness result. 
\begin{theorem}[Restricted approximate completeness] \label{th:approxCompl}
Let $S$ be a set of normal judgements in $\LLs(n)$ such that it has only
models valued over $[0, \infty)$. If a normal judgement%
\footnote{By an abuse of notation, here we actually mean that $\vdash \phi$ is
provably equivalent to a judgement in normal form.}
$\vdash \psi$ is a semantic consequence of $S$ in $\LLs(n)$, 
then for any $\e > 0$, $\e \vdash \psi$ is 
provable from $S$.
\end{theorem}

\begin{proof}
	We assume that $S$ is not finite (the finite case is covered by Theorem~\ref{completeness}) 
	and in normal form.
	%
We further assume that the constant $\un$ is never used neither in $S$ nor in $\psi$.
This will guarantee us to work on linear maps, rather than affine ones. The generalisation to the case of affine maps can be done by invoking the Fundamental Theorem of Linear Inequalities (for details see \eg~\cite[Corollary~7.1h]{Schrijver1998}).
	%
	Without loss of generality, as in Theorem~\ref{completeness}, we can assume that $\vdash\psi$ is provably equivalent to  
	\begin{equation*}
		(\bigotimes_{i \leq n} a_i*p_i ) \vdash (\bigotimes_{i \leq n} b_i*p_i )
	\end{equation*} 
	and the non-assertive judgements in $S$ are 
	\begin{equation*}
		\left\{
		\begin{aligned} 
			(\bigotimes_{i \leq n} a^1_i*p_i ) 
			&\vdash (\bigotimes_{i \leq n} b^1_i*p_i )
			\\
			(\bigotimes_{i \leq n} a^2_i*p_i ) 
			&\vdash (\bigotimes_{i \leq n} b^2_i*p_i ) 
			\\
			&\dots
		\end{aligned}
		\right.
	\end{equation*}
	Let $\X = \reals^{n}$ and $\Y = \reals^\naturals$, and denote by $\X^*$, $\Y^*$
	their dual spaces, respectively (\ie, $\X^*$ is the set of all continuous linear functionals $\X \to \reals$).
	Define the maps $T_S \colon \X \to \Y$
	and $t_\psi \colon \X \to \reals$, 
	for $x_i \in \reals$, $j \in \naturals$ as follows
	\begin{align*}
		T_S(x_1, \dots, x_n)(j) &\coloneq \sum_{i \leq n} (a_i^j - b_i^j) x_i 
		&&
		t_\psi(x_1, \dots, x_n) \coloneq \sum_{i \leq n} (a_i - b_i) x_i 
	\end{align*}
	In abstract terms, the set of judgements in $S$ can be thought of as
	the inequality $T_S(x_1, \dots, x_n) \geq 0$ and, similarly, $\psi$ as the 
	inequality $t_\psi(x_1, \dots, x_n) \geq 0$.
	Define the sets
	\begin{align*}
		V_S &= \{ x^* \in \X^* \mid \forall x \in \X.\; T_S(x) \geq 0 \text{ implies } x^*(x) \geq 0 \} \,,
		\\
		Z_S &= \{ x^* \in \X^* \mid \exists y^* \in \Y^*. \; x^* = T^*_S(y^*) \text{ and } y^* \geq 0 \} \,.
	\end{align*}
	where $T^*_S \colon \Y^* \to \X^*$ is the adjoint of $T_S$ 
	uniquely defined by the adjoint property as $T^*_S(y^*) \coloneq y^* \circ T_S$.
	
	By Hurwicz's general form of Farkas' Lemma~\cite{Hurwicz2014}, 
	we know that $V_S$ is the \emph{regularly convex envelope} of $Z_S$, which 
	corresponds to the topological closure $\overline{Z_S}$ of $Z_S$ for finite-dimensional vector spaces as $\X^*$ is. 
	Let $\pi_k \colon \Y \to \reals$ denote the $k^{\text{th}}$-projection function defined 
	as $\pi_k( (x_i)_{i \in \naturals} ) = x_k$. Clearly $\pi_k \in \Y^*$. Moreover, 
	the finite positive linear combinations of these projections, forms a dense
	subset $P$ of $\Y^*_+ = \{ y^* \mid y^* \geq 0 \}$. 
	Notice that $Z_S = T^*_S(\Y^*_+)$. Since $T^*_S$ is a continuous function, $\overline{Z_S} = \overline{T^*_S(P)}$.
	As we have established that $V_S = \overline{Z_S}$, any element of $V_S$
	can be approached arbitrary close by a map of the form $T^*_S(p)$,
	for $p \in P$. In simpler terms, as $T^*_S(p) = p \circ T_S$, any element of 
	$V_S$ is arbitrary close to a finite positive linear combination of judgements
	in $S$.
	
	By hypothesis, $\psi$ is a semantic 
	consequence of $S$, meaning that, $t_\psi \in V_S$.
	From this, to prove $\vdash \psi$ from $S$ in $\LLs$ we 
	just need to pick an appropriate $p \in P$ (which exists) 
	and replicate the finite positive linear combination 
	it represents by using the derived rule
	\begin{equation*}
		\infrule{\phi_1\vdash\psi_1 & \phi_2\vdash\psi_2
		}{r*\phi_1\ot s*\phi_2\vdash r*\psi_1\ot s*\psi_2}
	\end{equation*}
	to obtain a judgement $\vdash\psi'$ that is $\e$-close to $\vdash\psi$.
As we assumed $\psi$ to be provably equivalent to $\vdash \psi$,
by chaining the two proofs together, we are done.
\end{proof}

\gcut{\section{Encodings}
  
After having developed three propositional logics for the Lawvere quantale 
let us now ask how they relate to other logics, namely, the classical Boolean 
logic, {\L}ukasiewicz logic~\cite{Lukasiewicz30}, and Ben Yaacov's
continuous propositional logic~\cite{Yaacov13}.
 
 
\subsection{Encoding of Boolean Propositional Logic}
 
\begin{theorem}[Encoding Boolean Logic] \label{th:boolenc}
The theory $\T$ in $\LL$ axiomatized by the only axiom 
\begin{equation*}
(\textsc{tnd}) \quad
\vdash\phi\lor(\lnot\phi)
\tag{\it tertium non datur}
\end{equation*}
coincides with classical Boolean logic ($\Prop$ is the set of atomic propositions, 
and $\lor,\land,\lnot$ are the classic disjunction, conjunction and negation respectively). More exactly, a judgement of propositional logic (note that it also belongs to $\LL$) is provable in the classic Boolean logic iff it is provable in $\LL$ from $\T$.
\end{theorem}
\begin{proof}
Consider the judgement $\gamma$ in Boolean logic with atomic 
propositions in $\Prop$. This is, obviously, also a judgement in $\LL$.
	
Note that a model of Boolean logic is a map $\M \colon \Prop \to \{\textit{true},\textit{false}\}$ that extends, as standard, to all logical symbols. If we consider the bijection between $ \{\textit{true},\textit{false}\}$ and $\{0,\infty\}$ that maps $\textit{true}$ to $0$ and $\textit{false}$ to $\infty$, we discover that $\M$ is also a model of $\LL$ that satisfies the axiom (\textsc{tnd}), hence it is a model of $\T$. 
Similarly, any model of $\T$, because it is a model of (\textsc{tnd}), interprets all the formulas as $0$ or $\infty$ and it is not difficult to verify that it is a model of Boolean logic.
	
From the soundness of Boolean logic we know that if $\gamma$ is 
provable in Boolean logic, it is satisfied by all its models. But since the 
models of Boolean logic are the models of $\T$ and vice versa, we 
obtain that all the models of $\T$ are models of $\gamma$. 
Because $\T$ is finitely axiomatized, we can apply Theorem~\ref{completeness} 
and we get that $\gamma$ is provable in $\T$.
	
Similarly, if $\gamma$ is provable from $\T$, then in all the 
models of $\T$, $\gamma$ is satisfiable. But the models of $\T$ are 
exactly the models of Boolean logic. Hence, in Boolean 
logic we have that all the models satisfy $\gamma$. Further, the 
completeness of Boolean logic guarantees that $\gamma$ 
is provable in Boolean logic.
\end{proof}

\subsection{Encoding of {\L}ukasiewicz propositional logic}

Next we show that {\L}ukasiewicz logic (\L) can be encoded 
in $\LLun$. Before doing that, we briefly recall the syntax and 
semantics of $\L$-formulas and refer to \cite{HajekBook} for further details.

The formulas of \L\ are freely generated from the atomic 
propositions by three logical connectives:
\begin{equation*}
 \top \mid \lnot \phi \mid \phi \to \psi	
 \tag{{\L}ukasiewicz}
\end{equation*}
The models of $\L$ are assignments $w \colon \Prop \to [0,1]$ of the propositional 
symbols to the unit interval, which are uniquely extended to the formulas as
shown below%
\footnote{We consider the semantics used, \eg\ in~\cite{Yaacov13}, where $0$
corresponds to truth and any $r \in (0, 1]$ to a degree of truth/falsity, where $1$ is absolute falsity.} 
\begin{align*}
w(\top) &\coloneq 0 \\
w(\lnot \phi) &\coloneq 1 - w(\phi) \\
w(\phi \to \psi) &\coloneq \max\{ w(\phi) - w(\psi), 0 \} \,.
\end{align*}
A formula $\phi$ is satisfied by a model $m$ whenever $w(\phi) = 0$.

The following is an Hilbet-style axiomatisation for $\L$, 
\begin{align*}
(\text{A1}) \quad &(\phi \to \psi) \to \phi \\
(\text{A2}) \quad &((\theta \to \phi) \to (\theta \to \psi)) \to (\psi \to \phi) \\
(\text{A3}) \quad &(\phi \to (\phi \to \psi)) \to (\psi \to (\psi \to \phi)) \\
(\text{A4}) \quad &(\phi \to \psi ) \to ( \lnot \psi \to \lnot\phi )
\end{align*}
Deduction is defined in the natural way with \textit{modus ponens} being the only
rule of inference. This axiomatization is (weak) complete \wrt\ the semantics 
above~\cite{Chang58,RoseR58}.

The encoding function $e \colon \L \to \LLun$ mapping $\L$-formulas to 
$\LLun$-formulas is defined as follows, for $p \in \Prop$ and 
$\phi, \psi \in \L$
\begin{align*}
  e(p) &\coloneq p \lor \un &
  e(\lnot \phi) &\coloneq e(\phi) \lol \un
  \\
  e(\top) &\coloneq \top &
  e(\phi \to \psi) &\coloneq e(\psi) \lol e(\phi) \,.
\end{align*}

\begin{theorem}[Encoding {\L}ukasiewicz Logic] \label{th:lukaenc}
Let $\T$ be the theory in $\LLun$ axiomatized by the only axiom
\begin{equation*}
(\textsc{1-bbd}) \quad
\lnot\lnot \psi \vdash \un \lol \psi \,.
\end{equation*}
Then, a formula $\phi$ of {\L}ukasiewicz logic is provable 
in $\L$ iff the judgement $\vdash e(\phi)$ is provable in $\LLun$ from $\T$.
\end{theorem}
\begin{proof}
Note that, a model $\M$ in $\LLun$ satisfies (\textsc{1-bdd}) iff
\begin{equation}
  \text{for all $\phi \in \LLun$}\,, \quad
  \M(\phi) < \infty \; \Longrightarrow \; \M(\phi) \leq 1 \,.
  \label{bdd}
\end{equation}

($\Leftarrow$) Let $\phi \in \L$ and assume $\vdash e(\phi)$ is provable in
$\LLun$ from $\T$. Let $w \colon \Prop \to [0,1]$ be a model of $\L$.
Obviously, it is also a model for $\LLun$ that moreover 
satisfies \eqref{bdd}, hence $m \models_{\LLun} \T$. By Theorem~\ref{soundness}, 
$w$ satisfies the judgement $\vdash e(\phi)$, or equivalently $w(e(\phi)) = 0$.
By an easy induction on $\phi$, one can show that for any $\L$-model, 
$w(\phi) = w(e(\phi))$. Thus $w(\phi) = 0$. As the $\L$-model $w$ in the argument
above was generic, we just proved that $\phi$ is a tautology in $\L$. 
By Chang-completeness~\cite{Chang58} of $\L$, we then have that $\phi$ is
provable in $\L$.

($\Rightarrow$) Assume $\phi$ is provable in $\L$. 
Let $\M \colon \Prop \to \extreals$ be a model of $\T$. By an easy induction
on $\phi$, one can prove that, for any model $\M$ of $\T$, $\M(e(\phi)) \leq 1$
---this follows because for propositional variables $p \in \Prop$,
$\M(e(p)) = \min\{\M(p), 1 \} \leq 1$, and the semantical interpretation of the
other logical connectives is closed is well-defined in $[0,1]$. Thus $\M$ is 
a proper model for $\L$. As we assumed, $\phi$ is provable in $\L$, by 
the soundness of the axiomatization for $\L$, we have $\M(\phi) = 0$.
As before, $\M(\phi) = \M(e(\phi))$, Thus $\M(e(\phi)) = 0$. As the theory $\T$
is finitely axiomatizable, by Theorem~\ref{completeness}, $\vdash e(\phi)$
is provable from $\T$.
\end{proof}

\subsection{Encoding of continuous propositional logic}
At last, we consider the encoding in $\LLs$ of \emph{continuous propositional logic} (C\L), which is a conservative extension of {\L}ukasiewicz logic proposed by Ben Yaacov~\cite{Yaacov13} to reason about Banach spaces.

The syntax and semantics of C\L\ is that of $\L$, to which we add one single extra
unary logical connective $\Half\phi$ with the following semantical
interpretation in $[0,1]$:
\begin{equation*}
  w(\Half\phi) \coloneq \frac{1}{2} w(\phi) \,.
\end{equation*}

C\L\ has a (weak) complete axiomatization, consisting of the axioms (A1)--(A4) of 
$\L$ to which we add 
\begin{align*}
(\text{A5}) \quad &\Half\phi \to (\phi \to \Half\phi) \\
(\text{A6}) \quad &(\phi \to \Half\phi) \to \Half\phi
\end{align*}

For the encoding of C\L\ into $\LLs$, we consider as encoding function
$e \colon \text{C\L} \to \LLs$ the obvious extension of the one used for 
{\L}ukasiewicz logic and such that:
\begin{equation*}
  e(\Half\phi) \coloneq \frac{1}{2}*e(\phi) \,.
\end{equation*}

\begin{theorem}[Encoding Continuous Logic]
Let $\T$ be the theory in $\LLs$ axiomatized by the only axiom
\begin{equation*}
(\textsc{1-bbd}) \quad
\lnot\lnot \psi \vdash \un \lol \psi \,.
\end{equation*}
Then, a formula $\phi$ of continuous prop. logic is provable 
in C\L\ iff the judgement $\vdash e(\phi)$ is provable in $\LLs$ from $\T$.
\end{theorem}
\begin{proof}
Similar to the proof of Theorem~\ref{th:lukaenc}.
\end{proof}}

\section{Inference Systems for the Lawvere Quantale}
\label{sec:infSystems}

In this section, we extend further the concept of proof systems in \LLQ\  
and discuss \emph{inference systems}. These are obtained by requiring 
that a theory in $\LLs$ obeys extra inferences (proof rules), in addition to 
its axioms and to the inferences of $\LLs$. Because in \LLQ\ inferences 
cannot be internalized as judgements (see Fact~\ref{fact:wdt}), the effect 
of closing a theory by a rule does not always produce another theory, 
as happens in classical logics, but often a set of theories. 

Note that, all the inferences of \LLQ\ have finite sets of hypotheses 
and so, when one works with theories of \LLQ, will only have derived 
rules that contains a finite set of hypotheses. 
But if we now allow ourselves to work with inferences that might not be 
derived from the proof rules, we might have to handle inferences with 
a countable set of hypotheses. 
In fact, for our purpose, we are interested in only one type of such inferences 
that we shall call \emph{inductive inferences}.

\begin{definition}[Inductive inferences] \label{def:inductiveinf}
An \emph{inductive inference} in LLQ  is an inference of type 
\begin{equation*}
\infrule{\{\vdash\phi_i\mid i\in\mathbb N\}}{\vdash\psi}
\end{equation*}
such that for any $i,j\in\naturals$ with $i<j$, $\phi_j\vdash\phi_i$.
\end{definition}

Observe that the inferences with a finite number of hypotheses are 
all particular cases of inductive inferences, since they all can be equivalently
represented firstly as inferences with only one hypothesis (the conjunction of all the hypotheses); and secondly, this hypothesis can be seen as a constant sequence of hypotheses.
Hence, all the proof rules of \LLQ\  
and all the inferences that can be derived from them are inductive inferences.
Last but not least, observe that any axiom $\vdash\psi$ can be seen as an inductive inference with an empty set of hypotheses.

\begin{definition}
A theory $\T$ of $\LLs$ is \emph{closed under the inductive inference} 
\begin{equation*}
\infrule[I]{\{\vdash\phi_i\mid i\in\mathbb N\}}{\vdash\psi}
\end{equation*}
if for any diagrammatic extension $\T^+$ of $\T$ we have that 
\begin{itemize}
  \item either $\vdash\psi\in\T^+$, 
  \item or there exists $\e\in\preals$, $\e>0$ and 
	$i\in\mathbb N$ such that $\phi_i\vdash\e\in\T^+$.
\end{itemize}
\end{definition}
When a model $\M$ is such that  $\{\vdash\phi_i\mid i\in I\}\models_\M\vdash\psi$, we say that it is a model of the inference $I$.

Observe that the previous definition makes sense semantically, since it implies that a theory $\T$ of $\LLs$ is closed under the inference $I$ iff any model $\M$ of $\T$ is a model of $I$.

\begin{definition}[Inference System]
An \emph{inference system in $\LLs$} is a set $\R$ of inductive 
inferences in $\LLs$.

We say that a model satisfies an inference system if it is a model for 
each inference in the system. 
	
A theory is closed with respect to an inference system, 
if it is closed under each inference in the system.
\end{definition}

Since we can see any axiom of a theory as a particular inference with an empty set of hypotheses, we see that any finite
axiomatic system of \LLQ\ is, in fact, a particular case of an inference system. 
There exists, however, interesting mathematical theories, such as the quantitative equational logic, that cannot be presented by using an axiomatic system in \LLQ, but only by using an inference system.

Because the finite axiomatic systems in \LLQ\ are particular type of
inference systems, we can read the completeness theorem for 
finitely-axiomatized theories proven before, as a particular case of 
completeness for inference systems. In what follows, we will enforce 
these results and prove completeness results directly for inference systems. 

\begin{theorem}[Completeness for inference systems]
Let $\R$ be an inference system of $\LLs$ and 
\vspace*{-6pt}
\begin{equation*}
\infrule[I]{\{\vdash\phi_i\mid i\in\mathbb N\}}{\vdash\psi}
\end{equation*}
an inductive inference. If $I$  is satisfied by all the models of $\R$, 
then all the finitely axiomatized theories closed under $\R$ are 
also closed under $I$.
\end{theorem}

\begin{proof}
Let $\T$ be a finitely axiomatized theory closed under $\R$. 
Then any model $\M$ of $\T$ satisfies $I$, hence
\begin{itemize}
	\itemsep 0.1pt
	\item either $\M \models (\vdash \psi)$
	\item or there exist $i\in\mathbb N$ and $\e>0$ such 
	that $\M \models (\phi_i \vdash \e)$.
\end{itemize}
Applying Corollary \ref{characteristicmodel}, to $\M$ corresponds a diagrammatic theory $\T_\M$ such that $m \models (\phi \vdash \phi')$ implies $\phi\vdash\phi'\in\T_\M$; and to each diagrammatic extension $\T^+$ of $\T$ corresponds a model $\M_{\T^+}$ such that $\M_{\T^+} \models (\phi\vdash\phi')$ implies $\phi\vdash\phi'\in\T^+$.

Consider now an arbitrary diagrammatic extension $\T^+$ of $\T$. Since $\M_{\T^+}$ is a model of $\T$, we have that
\begin{itemize}
	\itemsep 0.1pt
	\item either $\M_{\T^+} \models (\vdash \psi)$, implying $\vdash\psi\in\T^+$,
	\item or there exist $i\in\mathbb N$ and $\e>0$ such that 
	$\M_{\T^+} \models (\phi_i \vdash \e)$, implying $\phi_i\vdash\e\in\T^+$.
\end{itemize}
Hence, $\T$ is closed under $I$.
\end{proof}

\subsection{The inference system of Quantitative Algebra}

In this section, we show how we can use \LLQ\ as a support for quantitative equational reasoning~\cite{Mardare16}.
Quantitative algebras~\cite{Mardare16} have been introduced, as a generalization of universal algebras, meant to axiomatize not only congruences, but algebraic structures over extended metric spaces. Given an algebraic similarity type $\Omega$,
a quantitative algebra is an $\Omega$-algebra supported by an extended metric space, so that all the algebraic operators are nonexpansive. Such a structure can be axiomatized using an extension of equational logics that uses, instead of equations of type $s=t$ for some terms $s,t$, quantitative equations of type $s=_\e t$ for some $\e\in\preals$. This quantitative equation is interpreted as "the distance between the interpretation of $s$ and $t$ is less or equal to $\e$".

In the theory of quantitative algebras, $=_\e$ are treated as classic Boolean predicates, so in any model, $s=_\e t$ is either \textit{true} or \textit{false}. However, a different way to look to this, is to actually think that we only have one equality predicate and $s=t$ is interpreted in the Lawvere quantale, thus allowing us to reason about the distance between $s$ and $t$. For instance, instead of $s=_\e t$, we could could use the syntax of $\LLs$, treate $s=t$ as an atomic proposition, and write $\e\vdash s=t$. This allows us to properly reason about extended metric spaces and encode, in our logic, the entire theory of quantitative equational reasoning.

In this section, we show how such an encoding is defined. $\LLs$ has already all the necessary ingredients to do this work. However, since $\LLs$ is only propositional, the way to do this is to treat all the equations as atomic propositions. This is exactly how we encode the classic equational logic into Boolean propositional logic. And, as in the classic case, while this is sufficient, it unfortunately requires an infinite set of axioms. 

As we have already anticipated, the theory of quantitative equational logic requires an inference system in $\LLs$, and it cannot be only encoded using an axiomatic system. This is because a judgement of type $s=_\e t\vdash s'=_\delta t'$ in quantitative reasoning corresponds to the inference $\dfrac{\e\vdash s=t}{\delta\vdash s'=t'}$ which cannot be internalized due to the fact that in $\LLs$ the deduction theorem fails.

Concretely, assuming an algebraic similarity type $\Omega$ and a set $X$ of variables, we construct all the possible algebraic terms. Let $\Omega X$ be the set of these terms. We define $\LLs$ for $\Prop=\{s=t\mid s,t\in\Omega X\}$. This gives us the syntax we need. 

The original axioms of quantitative equational logic are presented in Table~\ref{table:QA}, where they are stated for arbitrary terms $s,t,u, s_1,\dots,s_n,t_1,\dots,t_n\in\Omega X$, for arbitrary $n$-ary operator $f:n\in\Omega$, arbitrary positive reals $\e,\e'\in\preals$ and arbitrary decreasing convergent sequence $(\e_i)_{i\in\mathbb N}$ of positive reals with limit $\e$. These axioms, together with the standard substitution, cut and assumption rules, provide the proof system of quantitative equational logics. 

\begin{table}[tb]
	\begin{align*} 
		\text{\textbf{(Refl)}} \quad
		&  \vdash t=_0 t \quad \, \\ 
		\text{\textbf{(Symm)}} \quad 
		& s=_\e t\vdash t=_\e s \, \\ 
		\text{\textbf{(Triang)}} \quad
		&  t=_\e u, u=_{\e'}s\vdash t=_{\e+\e'}s \, \\ 
		\text{\textbf{(Max)}} \quad 
		& s=_\e t\vdash s=_{\e+\e'} t \, \\ 
		\text{\textbf{(Nexp)}} \quad
		& \{s_i=_\e t_i\mid i\leq n\}\vdash f(s_1,\dots,s_n)=_\e f(t_1,\dots,t_n) \quad \, \\ 
		\text{\textbf{(Cont)}} \quad 
		& \{s=_{\e_i} t\mid i\in\naturals\}\vdash s=_{\e}t \, 
	\end{align*}
	\caption{Quantitative algebras}
	\label{table:QA}
\end{table}

When translated into LLQ, the substitution, cut, and assumption rules are embedded in the way a proof operates. And the axioms of quantitative equational logic can be translated into the corresponding inferences in Table~\ref{table:qainf}.

\begin{table}[tb]
\begin{gather*} 
\infrule[refl]{}{\vdash t=t}
\qquad
\infrule[symm]{\e \vdash s= t}{\e \vdash t=s}
\qquad
\infrule[triang]{\e\vdash t=u & \e'\vdash u=s}{\e\ot\e'\vdash t=s}
\qquad
\infrule[max]{\e\vdash s=t}{\e\ot\e'\vdash s=t}
\\[2ex]
\infrule[nexp]{\e\vdash s_1=t_1 & \dots & \e\vdash s_n=t_n
	}{\e\vdash f(s_1, \dots ,s_n)=f(t_1,\dots,t_n)}
\qquad
\infrule[cont]{\e_1\vdash s=t & \dots & \e_i\vdash s=t & \dots}{\e\vdash s=t}
\end{gather*}
	\caption{Quantitative algebras}
	\label{table:qainf}
\end{table}

However, in this translation, the set of axioms is actually infinite, because the terms
equalities are names for atomic propositions and, as such, we will have, for instance, a
(\textsc{refl}) inference rule for each term $t$, a (\textsc{symm}) inference rule for
each tuple of terms $s$ and $t$, etc. This is not surprising, as the same situation
happens when we encode the classic equational logic developed for universal algebras in
propositional logic. 

Observe also that the axiom (\textsc{cont}) of quantitative equational logic, which is an infinitery axiom, is translated into $\LLs$ as an inductive inference rule. Indeed, first of all, we can convert each hypothesis of type $\e_i\vdash s=t$ into the equivalent one, $\vdash\e_i\lol(s=t)$. And secondly, because for $i\geq j$ we have $\e_i\leq\e_j$, and this implies that $\e_j\lol(s=t)\vdash\e_i\lol(s=t)$ is a theorem in $\LLs$.

A limitation of this encoding comes from the fact that we need an infinite set of
inductive inferences to rule the quantitative equational reasoning. But this is similar
with what is happening in the classical logic, when one encodes the classic equational
logic into propositional logic. And as in the classic case, this can be avoided by
extending $\LLs$ with predicates. This would allow us to present a more compact and
finitary inference system.  We leave this extension for future work.

\section{Conclusions}

In this paper, we developed three propositional logics interpreted in the Lawvere
quantale. We develop natural deduction systems for them, which collect rules similar to
rules well-known from other logics. We show that despite their natural arithmetic
interpretation, these logics manifest important metatheoretical original features that
differentiate them from other related logics. We prove that the logics are incomplete in
general, but complete if we restrict to finitely-axiomatized theories. We present a
normalization algorithm that proves the consequence is decidable and, when made efficient,
establishes complexity bounds.  Although one of our logics and results for it are known in
the context of product logics, our proofs are novel in all cases.

We also show that  quantitative equational logic can all be encoded in our new
settings. Moreover, we demonstrate that for this class of logics one can either use
axiomatic systems or systems of inferences, the second providing a higher expressivity
from the point of view of mathematical theories that can be developed. 

It would be interesting to extend our logic to allow products as well as sums. In that
case we would be interested in systems of polynomial inequalities, when the
Krivine--Stengle Positivstellensatz would surely come into play~\cite{K64,S74} in place of
the Farkas' Lemma. 

\section*{Acknowledgements}
Prakash Panangaden's visits to the University of Edinburgh have been supported by a
Strategic Talent Grant from the School of Informatics, University of Edinburgh under a
scheme funded by Huawei. His research has also been supported by research grants from
NSERC.  

\bibliographystyle{./entics}
\bibliography{main}

\end{document}

\newpage


\appendix

\title{Appendix}

This appendix contains three sections. Section \ref{A} contains some proofs of important
lemmas that were only stated in the paper, and a complete proof of the completeness
theorem. Section \ref{B} contains a set of results that characterize the deduction
aparatus in our logics and are used in the completeness proofs. Section
\ref{sec:normalizationAlg} contains a detailed presentation of the normalization
algorithm.


\section{Complete proofs of the main results}\label{A}

In this section, we present a couple of missing proofs of some lemmas.

\paragraph*{Proof of the Totality Lemma}
The following statements are provable in all \LLQ.
\begin{enumerate}[itemsep={1ex}, topsep={1ex}]
\item 
$\text{If }
\Big[
\dfrac{\vdash\rho~~~~\phi\vdash\psi}{\vdash\theta}
\text{ and }
\dfrac{\vdash\rho~~~~\psi\vdash\phi}{\vdash\theta}
\Big]
\text{ then }
\dfrac{\vdash\rho}{\vdash\theta}$.

\item 
$\text{If }
\Big[
\dfrac{\vdash\rho~~~~\vdash\lnot\phi}{\vdash\theta}
\text{ and }
\dfrac{\vdash\rho~~~~\vdash\lnot\lnot\phi}{\vdash\theta}
\Big]
\text{ then }
\dfrac{\vdash\rho}{\vdash\theta}$.
\end{enumerate}	

\begin{proof}(Proof of Lemma~\ref{totalityl1})	
	1. It is a consequence of the following after observing that 
	$\phi\vdash\psi$ and $\vdash\phi\lol\psi$ are equivalent. From
	\begin{equation*}
		\dfrac{\vdash\rho\land(\phi\lol\psi)}{\vdash\theta}~~~\text{ and }~~~\dfrac{\vdash\rho\land(\psi\lol\phi)}{\vdash\theta}
	\end{equation*}
	we derive, applying Lemma~\ref{meta4}, 
	$$\dfrac{\vdash(\rho\land(\phi\lol\psi))\lor(\rho\land(\psi\lol\phi))}{\vdash\theta}.$$ Hence, 
	$$\dfrac{\vdash\rho\land((\phi\lol\psi)\lor(\psi\lol\phi))}{\vdash\theta}.$$ 
	Involving (\textsc{tot}), we get $$\dfrac{\vdash\rho}{\vdash\theta}$$.
	
	2. It is a consequence of the following. From 
	\begin{equation*}
		\dfrac{\vdash\rho\land(\lnot\phi)}{\vdash\theta}~~~\text{ and }~~~\dfrac{\vdash\rho\land(\lnot\lnot\phi)}{\vdash\theta}
	\end{equation*}
	we derive, applying Lemma~ref{meta4}, 
	$$\dfrac{\vdash(\rho\land(\lnot\phi))\lor(\rho\land(\lnot\lnot\phi))}{\vdash\theta}.$$
	Hence, 
	$$\dfrac{\vdash\rho\land((\lnot\phi)\lor(\lnot\lnot\phi))}{\vdash\theta}.$$ 
	Involving (\textsc{wem}), we get 
	\begin{equation*}
		\dfrac{\vdash\rho}{\vdash\theta} \,.
	\end{equation*}
\end{proof}

\begin{proof}[Proof of Lemma \ref{max}]
	1. This is a consequence of soundness.  If a theory is inconsistent, then every formula
    of the form \(\varepsilon\vdash p\) as well as \(p\vdash\varepsilon\)can be proved so
    it is impossible to assign a value to formulas.
	\\2. For a model $\M$, define $\T_\M=\{\Gamma\vdash\phi\mid
    \Gamma\models_\M\phi\}$. Suppose $\vdash\phi\lor\psi\in \T_\M$.  
	Then either $\phi\vdash\psi$ or $\psi\vdash\phi$ are in $\T_\M$. 
	Assume $\phi\vdash\psi\in\T_\M$, then $\vdash \psi \lollol \phi\lor\psi\in\T_\M$. 
	Since $\vdash\phi\lor\psi\in \T_\M$, so $\vdash\psi\in\T_\M$.
\end{proof}

\paragraph*{Proof of Lemma~\ref{characteristictheory}}

In $\LLs$ we have that
\begin{enumerate}
  \item Every diagrammatic theory has a unique model. 
  \item Every model satisfies a unique diagrammatic theory.
  \item A theory is diagrammatic iff it is maximal consistent.
\item Every disjunctive consistent theory has a unique diagrammatic extension; and a unique model.
	
\end{enumerate}

\begin{proof} 

	1. Let $\T$ be a diagrammatic theory. For each $p\in\Prop$, there exists
    $v_p\in\extreals$ such that $p\dashv\vdash v_p\in\T$, where if $v_p=\infty$ the symbol
    $v_p$ is replaced by $\bot$ in the respective judgement. 
	
	We construct the model $\M_\T \colon \Prop \to \extreals$ by 
	$\M_\T(p)=v_p$ for any $p\in\Prop$.  Such an assignment of values to the propositional
    variables defines a unique model.  Let $\M'$ be an arbitrary model of $\T$. For 
	each $p\in\Prop$, $\M'(p)=v_p$, hence $\M'=\M_\T$.  Thus $\M_{\T}$ is the unique model
    satisfying $\T$.

	2. Let $\M$ be a model of $\LLs$. We showed in Lemma~\ref{max} that 
	$\T_\M=\{\Gamma\vdash\phi\mid\Gamma\models_\M\phi\}$
	is a consistent theory admitting $\M$ as a model. We prove it is 
	diagrammatic. Let $\phi\in\LLs$ and $v_\phi=\M(\phi)\in\extreals$. 
	If $v_\phi=\infty$, then $\phi\models_\M\bot$, hence $\phi\vdash\bot\in\T_\M$.
	Otherwise, $v_\phi\models_\M\phi$ and $\phi\models_\M v_\phi$, 
	implying $\phi\dashv\vdash v_\phi\in\T_\M$. 
	We prove that it is unique: suppose there exists a different 
	diagrammatic theory $\T\neq\T_\M$ that is satisfied by $\M$. 
	Let $\phi\vdash\psi$ be a judgement that is present in only one 
	of the theories. Suppose $\M(\phi)=r$ 
	and $\M(\psi)=s$. Since both these theories are diagrammatic, 
	$\phi\dashv\vdash r,~\psi\dashv\vdash s\in\T\cap\T_\M$. 
	Since $\phi\vdash\psi$ belongs to one of these theories, in this one 
	we also have $r\vdash s$. 
	Since this theory must be consistent, we obtain that $r\geq s$. 
	But then, $r\vdash s$ is a theorem in $\LLs$, hence is present 
	in the two theories. But then, using the fact that 
	$\phi\dashv\vdash r,~\psi\dashv\vdash s\in\T\cap\T_\M$, we get
    $\phi\vdash\psi\in\T\cap\T_\M$. Hence, the assumption 	that exists a judgement
    $\phi\vdash\psi$ that belongs to only one 	of the theories is false, implying that the
    two theories coincide.

	
	
	3.	($\Rightarrow$) Suppose $\T$ is diagrammatic and let $\phi\vdash\psi\not\in\T$. 
	Then $\M_\T(\phi)<\M_\T(\psi)$. Hence, there exists
	$r,s\in\preals$ such that $\M_\T(\phi)<r<s<\M_\T(\psi)$. As 
	$\T$ is diagrammatic and disjunctive, $r\vdash\phi,~\psi\vdash s\in\T$. 
	From here we get that in the theory generated by 
	$\T\cup\{\phi\vdash\psi\}$ we can prove $r\vdash s$ for 
	$r<s$, which in $\LLs$ is sufficient to prove inconsistency.
	\\($\Leftarrow$) Suppose $\T$ is maximal consistent. We firstly prove that it is disjunctive. Suppose $\vdash\phi\lor\psi\in\T$ but $\vdash\phi\not\in\T$. Then, from the maximality, there must exists a proof for $\vdash\bot$ in $\T\cup\{\vdash\phi\}$. Similarly, if $\vdash\psi\not\in\T$, there must exists a proof for $\vdash\bot$ in $\T\cup\{\vdash\psi\}$. Combining these two proofs, one can get a proof of $\vdash\bot$ from $\T\cup\{\vdash\phi\lor\psi\}=\T$ - contradictions with the consistency of $\T$. Hence $\T$ is disjunctive and consistent. Hence, it is diagrammatic.
	
	4. Consider a disjunctive consistent theory $\T$ and let $p\in\mathbb P$. As it is disjunctive, for any $r\geq 0$ we can either prove $p\vdash r$ or $r\vdash p$. This determines a unique real $v_p$ such that if $p\vdash r\in\T$, then $v_p\geq r$, and similarly if $r\vdash p$. Thus, we have a model of $\T$. Further, we use the other cases.
\end{proof}

\begin{proof} (Complete proof of Theorem~\ref{completeness})
	
	The proof is made similarly for $\LL,\LLun$ and $\LLs$, adapting 
	only the arguments to the appropriate context. Hereafter we 
	present the proof for $\LLs$, as it is the most complex of all.
	
	We do the proof in three steps.
	\medskip
	
	I. Firstly, we prove the statement under the assumption that $\gamma$ is a normal judgement and $\T$ a normal theory. There are a few cases we have to consider.
	
	1. If $\gamma$ is the alethic judgement $\vdash\lnot p$ for some 
	atomic proposition $p\in\mathbb{P}$. In any model $\M$ of this 
	judgement we have $\M(p)=\infty$. Hence, the hypothesis 
	guarantees that in all models of $\T$, the value of $p$ is $\infty$. 
	This means that $p$ must be present in the axiomatization, because 
	otherwise it could have any value. Since $\T$ has a normal 
	axiomatization, it must have at least one assertive judgement of $p$. 
	
	If it is $\top\vdash p$, then all the models of $\T$ will evaluate $p$ to 
	$0$  ---impossible.
	
	If it is the finitist judgement $\vdash\lnot\lnot p$, all the models of $\T$ will evaluate $p$ to a finite value ---impossible.
	
	The only possibility left is $\vdash\lnot p$, and in this case our judgment is indeed present in $\T$. 
	
	2.  If $\gamma$ is the finitist judgement $\vdash\lnot\lnot p$ for some atomic proposition $p\in\mathbb{P}$. In any model $\M$ of this judgement we have that the value of $p$ is finite. Hence, the hypotehsis guarantees that in all models of $\T$, the value of $p$ is finite as well. This means firstly that $p$ must be present in the axiomatization, because otherwise it could have $\infty$ as value. Since $\T$ has a normal axiomatization, it must have at least one assertive judgement of $p$. 
	
	If it is $\top\vdash p$, then it also contains $\vdash\lnot\lnot p$, since in all LLQ  we can prove $$\dfrac{\vdash \phi}{\vdash\lnot\lnot\phi}.$$
	
	If it is the finitist judgement $\vdash\lnot\lnot p$ itself, then $\gamma$ is in $\T$.
	
	If it is $\vdash\lnot p$, in this case all the models of $\T$ assign to $p$ the value $\infty$ ---impossible. 
	
	3. The last case is when $\gamma$ is affine:
	\begin{equation*}
		(\bigotimes_{i \leq n} r_i*p_i ) \ot r*\un \vdash (\bigotimes_{j \leq m} s_j*q_j ) \ot s*\un
	\end{equation*}  
	for $r,s,r_i,s_j \in\preals$, $p_i, q_j \in \Prop$, with $m$ and/or $n$ possibly having  value $0$. We need to prove that our judgement is provable from the axioms of $\T$.
	
	Consider a normal axiomatization of $\T$. From the definition of normal axiomatization, some of these axioms are alethic and none of the atomic proposition present in these alethic axioms are present in any other axioms. 
	
	If some of the atomic propositions $p_i,q_j$ appear in the alethic axioms of $\T$. There are a few cases to consider:
	\begin{itemize}
		\item If the alethic axioms where these appear are of type $\top\vdash p$ for some atomic proposition $p$, then it is sufficient to cancel all instances of $p$ from our judgement and focus on proving that the new judgement can be proven from the axioms of $\T$. This is sufficient to claim that our initial judgement can be proven as well, since the following rules are provable in all \LLQ:
		$$\dfrac{\top\vdash p~~~\phi\vdash\psi}{\phi\ot r*p\vdash\psi} ~~~\text{ and }~~~\dfrac{\top\vdash p~~~\phi\vdash\psi}{\phi\vdash\psi\ot r*p}.$$
		\item If an alethic axiom of type $\vdash\lnot p_i$ is in $\T$, then 
		\begin{equation*}
			(\bigotimes_{i \leq n} r_i*p_i ) \ot r*\un \vdash (\bigotimes_{j \leq m} s_j*q_j ) \ot s*\un
		\end{equation*}  
		is provable in $\T$ because in any \LLQ the following rule is 
		provable: $$\dfrac{p\vdash\bot}{r*p\ot\phi\vdash\psi}.$$
		\item If an alethic axiom of type $\vdash\lnot q_i$ is in $\T$, then all the models of $\T$ evaluate $s_1*q_1\ot..\ot s_n*q_n\ot s*\un$ to $\infty$. Since all the models of $\T$ are models of 
		$\gamma$, we must have that all models of $\T$ evaluate $r_1*p_1\ot..\ot r_n*p_n$ to $\infty$. Hence, all the models of $\T$ evaluate at least one of the atomic propositions $p_1,..,p_n$ to $\infty$. But for each $p_i$, in the normal axiomatization of $\T$ there is either the axiom $\vdash\lnot p_i$, or the axiom $\vdash\lnot\lnot p_i$, from the definition of normal axiomatization. If for all $p_i$ $\T$ contains the finitary axioms, then no model of $\T$ evaluates $r_1*p_1\ot..\ot r_n*p_n$ to $\infty$ - impossible. Hence, at least one axiom $\vdash\lnot p_i$ must be in the axiomatization of $\T$ - and we are in the previous case.
	\end{itemize}

	It remains to establish the case when none of the atomic propositions $p_1,\dots, p_n,
    q_1,\dots,q_m$ appear in the alethic axioms of $\T$. Consequently, all these atomic
    propositions apear in finitist axioms of $\T$. Hence, we are looking after real values
    for these atomic propositions. 
	
	In this case, we will prove that our judgement 
	\begin{equation*}
		(\bigotimes_{i \leq n} r_i*p_i ) \ot r*\un \vdash (\bigotimes_{j \leq m} s_j*q_j ) \ot s*\un
	\end{equation*} 
	can be proven from the non-assertive axioms of $\T$ only. 
	
	Using commutativity and associativity of tensor product, we will reorganize both our
    judgement and the non-assertive axioms of $\T$ so that we put together different
    copies of the same atomic proposition in a tensor product and use the facts that
    $0p=\top$ and $r*\top\ot\phi\dashv\vdash\phi$.

	%
	%
	%
	%
	
	So, without losing generality, we can assume that our judgement $\gamma$ is  
	\begin{equation*}
		(\bigotimes_{i \leq k} a_i*x_i ) \ot r*\un \vdash (\bigotimes_{i \leq k} b_i*x_i ) \ot s*\un
	\end{equation*} 
	and the non-assertive axioms of $\T$ are 
	\begin{equation*}
		\left\{
		\begin{aligned} 
			(\bigotimes_{i \leq k} a^1_i*x_i ) \ot r^1*\un 
			&\vdash (\bigotimes_{i \leq k} b^1_i*x_i ) \ot s^1*\un
			\\
			&\dots
			\\
			(\bigotimes_{i \leq k} a^l_i*x_i ) \ot r^l*\un 
			&\vdash (\bigotimes_{i \leq k} b^l_i*x_i ) \ot s^l*\un
		\end{aligned}
		\right.
	\end{equation*}	
	for some positive reals $a_i, b_i, a_i^j, b_i^j, r, s, r^j, s^j$ and atomic propositions $x_1, \dots, x_k$.
	Consider the matrices $A \in \reals^{l \times k}$, 
	$C \in \reals^{k\times 1}$ and vector $\beta \in \reals^{k}$.
	\begin{align*}
		A&=\begin{pmatrix} 
			a^1_1-b^1_1 & \dots & 
			a^1_k-b^1_k \\
			\dots &\dots&\dots
			\\
			a^l_1-b^l_1 &\dots& a^l_k-b^l_k 
		\end{pmatrix}
		&&
		\beta =\begin{pmatrix} r^1-s^1\\ \dots \\r^l-s^l\end{pmatrix}
		\\
		C &= (b_1 - a_1, \dots, b_k - a_k) \,
	\end{align*}
	and let $\delta = r-s$.
	Since from the hypothesis, any model of $\T$ is a model of $\gamma$, being also our working hypothesis in this case, there exists no $x = (x_1,\dots,x_k)\in \mathbb R^k$ 
	such that 
	\begin{gather*}
		Ax + \beta \geq 0 \,, \quad Cx + \delta >0 \,.
	\end{gather*}
	
	(Integer case)
	If we are in the case of $\LL$ or $\LLun$, then $a_i,b_i,a^j_i,b^j_i$ are all positive integers, and by applying Mozkin's rational transposition problem~\cite{Motzkin51} (here adapted already for the affine case) there exist 
	$t_0 \in \mathbb{Z}$ and $t=(t_1,\dots,t_l)\in\mathbb{Z}^{l\times 1}$, such that
	\begin{gather*}
		C x + \delta = t(Ax + \beta) + t_0 \,, \quad t \geq 0 \,, \quad t_0 \geq 0 \,.
	\end{gather*}
	
	But this means that by considering $|t_i|$ copies of the $i$-th non-assertive axiom
	of $\T$, and $t_0$ copies of $\un$, and repeatedly apply the rule 
	$$\dfrac{\phi_1\vdash\psi_1~~~\phi_2\vdash\psi_2}{\phi_1\ot\phi_2\vdash\psi_1\ot\psi_2}$$
	we obtain a proof in $\T$ for
	\begin{equation*}
		\big( (\bigotimes_{i \leq k} a_i*x_i ) \ot r*\un \big) 
		\vdash 
		\big( (\bigotimes_{i \leq k} b_i*x_i ) \ot s*\un \big) \,.
	\end{equation*}
	
	(Real case)
	If we are working in $\LLs$, by the general Mozkin transposition theorem there exist 
	$t_0 \in \reals$ and $t=(t_1,\dots,t_l)\in\reals^{l\times 1}$, such that
	\begin{gather*}
		C x + \delta = t(Ax + \beta) + t_0 \,, \quad t \geq 0 \,, \quad t_0 \geq 0 \,.
	\end{gather*}
	In $\LLs$ we can repeatedly apply the derived rule
	\begin{equation*}
		\infrule{\phi_1\vdash\psi_1 & \phi_2\vdash\psi_2
		}{r*\phi_1\ot s*\phi_2\vdash r*\psi_1\ot s*\psi_2}
	\end{equation*}
	and by following the pattern from $t$, $t_0$ and obtain 
	a proof from $\T$ for 
	\begin{equation*}
		\big( (\bigotimes_{i \leq k} a_i*x_i ) \ot r*\un \big) 
		\vdash 
		\big( (\bigotimes_{i \leq k} b_i*x_i ) \ot s*\un \big) \,.
	\end{equation*}	
	
	II. The second step in the proof is to assume that $\T$ is finitely-axiomatized 
	but not necessarily normal, and $\gamma$ is in normal form. 
	Let $V$ be a finite axiomatization of $\T$ and let $\N(V)=\{V_1,\dots,V_n\}$. 
	Hence, by applying Theorem~\ref{normalrep}, any model of $\T$ is a model 
	for at least one theory axiomatized by $V_i$; and reverse, any model of 
	any theory axiomatized by $V_i$ is a model of $\T$.
	
	Since any model of $\T$ is a model of $\gamma$, we get that for any 
	$i\leq n$, any model of $V_i$ is a model of $\gamma$. But $V_i$ is a normal axiomatization, so we can apply step I. and obtain that 
	$$\text{ for all }~~i\leq n,~~~\dfrac{V_i}{\gamma}$$
	Observe that $\{V_1,\dots, V_n\}$ is the signature of a decision tree with 
	the root labelled by $V$. Applying the second totality Lemma~\ref{totalityl2}, 
	we get 
	$$\dfrac{V}{\gamma}$$ 
	hence, $\gamma$ is provable in $\T$.
	
	III. Consider now the case when neither $\T$, nor $\gamma$ are in normal 
	form. Suppose that $V$ is a finite axiomatization for $\T$, that 
	$\N(V)=\{V_1,\dots,V_n\}$, that $\N(\gamma)=\{U_1,\dots,U_m\}$, 
	and assume that the trace of the path from the root to $U_i$ in the 
	decision tree $\N(\gamma)$ is $L_i$.
	
	Let $i\leq m$. If $\M$ is a model for $V\cup L_i$, since all the models of 
	$V$ are models of $\gamma$, we get that $\M$ is a model for 
	$L_i\cup\{\gamma\}$. But because $\N(\gamma)$ is a decision tree we 
	know that 
	$$\dfrac{\gamma~~L_i}{U_i}.$$ 
	Hence, $\M$ is a model of $U_i$. 
	
	Consider an arbitrary judgement $\Delta\vdash\delta\in U_i$. 
	It is in normal form and, moreover, we just proven that any model 
	of $V\cup L_i$ is a model of $\Delta\vdash\delta$. 
	We can apply the case II. and get that 
	$$\dfrac{V~~L_i}{\Delta\vdash\delta}~~\text{ implying further that }~~\dfrac{V~~L_i}{U_i}.$$
	
	But from the construction of the decision tree $\N(\gamma)$ 
	we know that $\dfrac{U_i}{\gamma}$.
	
	Hence, for each $i\leq m$, 
	$$\dfrac{V~~L_i}{\gamma}.$$ 
	And applying the second totality Lemma~\ref{totalityl2}, 
	we get $\dfrac{V}{\gamma}.$
\end{proof}


\section{Totality Lemmas and Decision Trees}\label{B}
In this section, we prove a series of metatheorems of \LLQ\ that have been applied in some of the main results. 

\begin{lemma}\label{meta1}
	In all \LLQ, 
	\begin{equation*}
		\text{If} \quad 
		\dfrac{\vdash\phi}{\vdash\psi} 
		\quad \text{then} \quad
		\left[ \dfrac{\vdash\phi\land\theta}{\vdash\psi\land\theta} 
		\quad\text{and}\quad
		\dfrac{\vdash\phi\lor\theta}{\vdash\psi\lor\theta}
		\right] \,.
	\end{equation*}
\end{lemma}

\begin{proof}
	We show one case.
	\begin{equation*}
		\infrule[$\land_2$]{
			\infrule{
				\infrule[$\land_3^*$]{\vdash \phi \land \theta}{\vdash \phi}
			}{\vdash \psi}
			&
			\infrule[$\land_3$]{\vdash \phi \land \theta}{\vdash \theta}
		}{\vdash \psi \land \theta}
	\end{equation*} 
	where ($\land_3^*$) is the symmetric rule of ($\land_3$), which is derivable.
	The other case follows similarly by using ($\lor_1$), ($\lor_2$), ($\lor_2^*$).
\end{proof}

\begin{lemma}\label{meta2}
	In all \LLQ
	\begin{equation*}
		\dfrac{\vdash\theta}{\vdash\phi} 
		\quad \text{ iff } \quad
		\dfrac{\vdash\theta}{\vdash\phi\land\theta}
	\end{equation*}
\end{lemma}
\begin{proof}
	($\Leftarrow$): $$\dfrac{\dfrac{\vdash\theta}{\vdash\phi\land\theta}}{\vdash\phi}.$$
	($\Rightarrow$): $\dfrac{\vdash\theta}{\vdash\phi}$ implies, using lemma \ref{meta1}, $\dfrac{\theta\land\theta}{\theta\land\phi}$, which is equivalent to $\dfrac{\vdash\theta}{\vdash\phi\land\theta}$. 
\end{proof}

\begin{lemma}\label{meta3}
	In all LLQ ,
	$$\left[ ~\dfrac{\vdash\theta}{\vdash\phi}~~\text{ and }~~\dfrac{\vdash\theta}{\vdash\psi}~\right]~~~\text{ iff }~~~\dfrac{\vdash\theta}{\vdash\phi\land\psi}.$$
\end{lemma}

\begin{proof}
	($\Leftarrow$): $$\dfrac{\dfrac{\vdash\theta}{\vdash\phi\land\psi}}{\vdash\theta}.$$
	($\Rightarrow$): Using lemma \ref{meta1}, from $\dfrac{\vdash\theta}{\vdash\phi}$ we infer $\dfrac{\vdash\theta\land\psi}{\vdash\phi\land\psi}$. Similarly, applying lemma \ref{meta2}, from $\dfrac{\vdash\theta}{\vdash\psi}$ we derive $\dfrac{\vdash\theta}{\vdash\theta\land\psi}$. Next, modus ponens completes the proof. 
\end{proof}

\begin{lemma}\label{meta4}
	In all LLQ ,
	$$\text{If }~~\left[~\dfrac{\vdash\phi}{\vdash\theta}~~~\text{ and }~~~\dfrac{\vdash\psi}{\vdash\theta}\right]~~~~\text{ then }~~~~\dfrac{\vdash\phi\lor\psi}{\vdash\theta}.$$	
\end{lemma}

\begin{proof}
	Applying lemma \ref{meta1} to $~~\dfrac{\vdash\phi}{\vdash\theta}~~$ and to $~~\dfrac{\vdash\psi}{\vdash\theta}~~$ we obtain $~~\dfrac{\vdash\phi\lor\psi}{\vdash\theta\lor\psi}~~$ and $~~\dfrac{\vdash\phi\lor\psi}{\vdash\theta\lor\phi}~~$ respectively. Using them in the context of lemma \ref{meta3} we get
	$$\dfrac{\vdash\phi\lor\psi}{\vdash(\theta\lor\phi)\land(\theta\lor\psi)},$$ which implies further $$\dfrac{\vdash\phi\lor\psi}{\vdash\theta\land(\phi\lor\psi)}.$$ Now, applying Lemma~\ref{meta3}, we get 
	\begin{equation*}
		\dfrac{\vdash\phi\lor\psi}{\vdash\theta} \,.
	\end{equation*}
\end{proof}

Now we prove a second version, a bit stronger, of the totality lemma that will help us with the normalization algorithm.

Hereafter, if $S,R,V$ are sets of judgements, we write 
$$\dfrac{S}{~V~}$$ to denote the fact that for all the judgements 
$\gamma \in V$, we have 
$$\dfrac{S}{~\gamma~}$$

There exists a stronger version of totality lemma that requires 
the concept of \emph{decision tree}.
\begin{definition}[Decision tree]
A \emph{decision tree} in $\mathcal{L}$, where 
$\mathcal{L}\in\{\LL,\LLun,\LLs\}$, is a finite binary tree with the nodes 
labelled byfinite sets of judgements in $\mathcal{L}$ such that
\begin{itemize}
  \item when a node labelled by $R$ has only one child, labelled by $V$, then the following rules are provable in $\mathcal{L}$ 
$$\dfrac{~R~}{V}~~~\text { and }~~~\dfrac{~V~}{R}.$$

  \item when a node labelled by $R$ has two children, the labels 
  of the children are of type $V_1\cup\{\vdash\phi\}$ and 
  $V_2\cup\{\vdash\psi\}$ respectively, where $(\vdash\phi ,\vdash\psi)$ 
  is a pair of supplementary judgements, 
  and the following rules are provable in $\mathcal{L}$ 
  \begin{align*}
  \dfrac{V_1 \quad \vdash\phi}{R} \quad \dfrac{V_2 \quad \vdash\psi}{R}
  \quad
  \dfrac{R \quad \vdash\phi}{V_1}
  \quad \text{ and }\quad
  \dfrac{R \quad \vdash\psi}{V_2} \,.
  \end{align*}
\end{itemize}
The \emph{signature of a decision tree} is the set of the labels of its leafs.
If $F_1,\dots, F_n$ are all the nodes that split on a path between two connected nodes $S$ and $T$, and for each $i\leq n$, $\vdash\phi_i$ is the judgement from the pair of  supplementary judgements used to split the node $F_i$, we call the set $\{\vdash\phi_1, \dots, \vdash\phi_n\}$ the \emph{trace of the path} from $S$ to $T$.
\end{definition}

Decision trees can be composed: when the label of a leaf of one 
tree coincides with the label of the root of another, we can merge the 
root of the second with the leaf of the first obtaining a composed tree. 
It's trivial to verify that this tree is indeed a decision tree.  

\begin{corollary}\label{dectree}
Let $S$ and $T$ be the labels of two connected nodes of a 
decision tree in $\mathcal{L}$, where $\mathcal{L}\in\{\LL,\LLun,\LLs\}$, 
and $V$ be the trace of the path from $S$ to $T$. 
Then, the following rules are provable in $\mathcal{L}$.
$$\dfrac{S~~V}{T}~~\text{ and }~~\dfrac{T~~V}{S} \,.$$
\end{corollary}

\begin{lemma}[Second Totality Lemma]\label{totalityl2}
Consider a decision tree in $\mathcal{L}$,  of signature $\{F_1,\dots ,F_n\}$ 
and the root labelled by $R$, where $\mathcal{L}\in\{\LL,\LLun,\LLs\}$. 
Suppose that for each $i\leq n$, let $L_i$ is the trace of the path from 
the root to the leaf $F_i$. The following implications are provable 
in $\mathcal{L}$.
\begin{enumerate}
	\item $$\text{If }~~\left[\text{ for all }i\leq n,~~\dfrac{R~~L_i}{\vdash\theta}~~~\right]~~~~\text{ then }~~~~\dfrac{R}{\vdash\theta};$$
	\item $$\text{If }~~\left[\text{ for all }i\leq n,~~\dfrac{F_i}{\vdash\theta}~~~\right]~~~~\text{ then }~~~~\dfrac{R}{\vdash\theta}.$$
\end{enumerate}

\end{lemma}

\begin{proof}
1. We prove the statement by induction on the structure of the decision tree, starting from the leafs upwards. The base case are the trees with only one level of children.

If $F$ is a leaf that has no siblings and $S$ is its parent, then the trace of the path from $S$ to $F$ is empty and the implication is trivially true.

If $F_1, F_2$ are two siblings leafs having the parent $S$, and $\vdash\phi_1$, $\vdash\phi_2$ are the  supplementary judgements marking the spliting to $F_1$ and $F_2$ respectively, then the trace of the path from $S$ to $F_i$ is $\{\vdash\phi_i\}$. In this case, the implication is one of the cases 2. or 4. of the first totality lemma \ref{totalityl1}.

For the inductive step, suppose the statement is true for trees of maximum depth $n$, and we prove it for trees of depth $n+1$. 

Let $S$ be the root of a tree of depth $n+1$.

If $S$ has a unique child $F$, the statement is trivially true since the tree with the root $F$ has the length $n$ and the trace of the root from $S$ to $F$ is empty.

If $S$ has two children, $S_1$ and $S_2$, these are the roots of two trees of depth at most $n$. Suppose that the first one has the root $S_1$, the signature 
$\{F^1_1,\dots, F^1_{k_1}\}$ and the trace of the path from $S_1$ to $F^1_i$ is $L^1_i$. And the second has the root $S_2$, the signature 
$\{F^2_1,\dots,F^2_{k_2}\}$ and the trace of the path from $S_2$ to $F^2_i$ is $L^2_i$. Suppose also that the split in $S$ is marked by two  supplementary judgements, $\vdash\phi_1$ leading to $S_1$ and $\vdash\phi_2$ leading to 
$S_2$. 

Consider now two new decision trees, $T_1$ obtain by adding 
$\vdash\phi_1$ to all the nodes of $S_1$ except the root and having 
$S\cup\{\vdash\phi_1\}$ as a root; and $T_2$ obtain by adding 
$\vdash\phi_2$ to all the nodes of $S_2$ except the root and having 
$S\cup\{\vdash\phi_2\}$ as a root. Because $S_1$ and $S_2$ are 
children of $S$ in the initial tree, guarded by $\vdash\phi_1$ and 
$\vdash\phi_2$ respectively, we get that indeed $T_1$ and $T_2$ 
are decision trees. Moreover, the trace of the path from the root of 
$T_i$ to $G^i_j$ is still $L^i_j$ as in the case of $S_i$. Moreover, 
$T_i$ has the depth at most $n$. 

From the working hypotehsis we know that for each $i\leq 2$ we 
have that for each $j\leq k_i$, 
$$\dfrac{S~~\vdash\phi_i~~L^i_j}{\vdash\theta}.$$ 
Hence, for each $i\leq 2$ we have that for each $j\leq k_i$, 
$$\dfrac{S\cup\{\vdash\phi_i\}~~L^i_j}{\vdash\theta}.$$ 
Applying the inductive hypotehsis in the tree $T_i$ we get from here that
 $$\dfrac{S~~\vdash\phi_1}{\vdash\theta}~~\text{ and }~~\dfrac{S~~\vdash\phi_2}{\vdash\theta}.$$ 
 Next, the cases 2. or 4. of the first totality lemma \ref{totalityl1}, we get the proof.

2. We prove by induction on the structure of the decision tree, starting from leafs upwards, that $\vdash\theta$ is provable from each set of judgements that labels a node of the tree. It implies that it is provable from $R$. 

If $F$ is a leaf that has no siblings and $S$ is its parent, then from $\dfrac{F}{\vdash\theta}$ we derive $\dfrac{S}{\vdash\theta}$ since we have $\dfrac{~S~}{F}$.

If $F_1, F_2$ are two siblings leafs having the parent $S$, then there exists a pair of  supplementary judgements $\vdash\phi$ and $\vdash\psi$ such that $F_1=V_1\cup\{\vdash\phi\}$ and $F_2=V_2\cup\{\vdash\psi\}$. From the definition of decision tree, we know that $$\dfrac{S~~~\vdash\phi}{V_1}~~~\text{ and }~~~\dfrac{S~~~\vdash\psi}{V_2},$$ implying $$\dfrac{S~~~\vdash\phi}{V_1~~~\vdash\phi}~~~\text{ and }~~~\dfrac{S~~~\vdash\psi}{V_2~~~\vdash\psi}.$$ Using the hypotehsis, we have $$\dfrac{V_1~~~\vdash\phi}{\vdash\theta}~~~\text{ and }~~~\dfrac{V_2~~~\vdash\psi}{\vdash\theta}.$$ From these we get $$\dfrac{S~~~\vdash\phi}{\vdash\theta}~~~\text{ and }~~~\dfrac{S~~~\vdash\psi}{\vdash\theta}.$$ Since $\vdash\phi$ and $\vdash\psi$ are  supplementary judgements, using the first totality lemma \ref{totalityl1}, we get $\dfrac{S}{\vdash\theta}.$
\end{proof}


\section{Normalization Algorithm}\label{sec:normalizationAlg}

In this section, we will prove that any finitely axiomatized theory has a normal representation, and moreover, we will propose an algorithm that will construct a normal representation for such a theory.

\paragraph{The discrimination function} $\D$ is a nondeterministic function that takes a judgement $\Gamma\vdash\phi$ as input and returns a finite set of finite sets of judgements as output, so that $M$ is a model of $\Gamma\vdash\phi$ iff it is a model for at least one of the sets of judgements in $\D(\Gamma\vdash\phi)$. It is defined inductively on the syntax of the judgement, and the cases are described below. To simplify the future development, we present the cases as decision trees (recall Definition~\ref{decisiontree}) - it is not difficult to verify that each of the graph belows are indeed decision trees. The function $\D$ takes the label of the root as input and returns the signature of the tree as output.  

\smallskip
Rule 1:\smallskip

\begin{tikzpicture}
	[
	grow                    = right,
	sibling distance        = 3em,
	level distance          = 10em,
	edge from parent/.style = {draw, -latex},
	every node/.style       = {font=\footnotesize},
	sloped
	]
	\node [root] {$\phi_1,\dots,\phi_n\vdash\psi$}
	child { node [env] {$\phi_1\ot\dots\ot\phi_n\vdash\psi$}
		edge from parent node [above] {}};
\end{tikzpicture}

Rule 2:\smallskip

\begin{tikzpicture}
	[
	grow                    = right,
	sibling distance        = 3em,
	level distance          = 10em,
	edge from parent/.style = {draw, -latex},
	every node/.style       = {font=\footnotesize},
	sloped
	]
	\node [root] {$\gamma\vdash (\phi\land\psi)\ot\rho$}
	child { node [env] {$\gamma\vdash \phi\ot\rho$ \\ $\gamma\vdash \psi\ot\rho$}
	edge from parent node [above] {}};
\end{tikzpicture}

Rule 3:\smallskip

\begin{tikzpicture}
	[
	grow                    = right,
	sibling distance        = 3em,
	level distance          = 10em,
	edge from parent/.style = {draw, -latex},
	every node/.style       = {font=\footnotesize},
	sloped
	]
	\node [root] {$\gamma\vdash (\phi\lor\psi)\ot\rho$}
	child { node [env] {$\phi\vdash\psi$\\$\gamma\vdash \psi\ot\rho$}
		edge from parent node [below] {} }
	child { node [env] {$\psi\vdash\phi$\\$\gamma\vdash \phi\ot\rho$}
		edge from parent node [above] {} };
\end{tikzpicture}

Rule 4: (For $r>0$) \smallskip

\begin{tikzpicture}
	[
	grow                    = right,
	sibling distance        = 3em,
	level distance          = 10em,
	edge from parent/.style = {draw, -latex},
	every node/.style       = {font=\footnotesize},
	sloped
	]
	\node [root] {$\gamma\vdash r*\bot\ot\rho$}
	child { node [env] {$\gamma\vdash\bot$}
	edge from parent node [above] {}};
\end{tikzpicture}

Rule 5:\smallskip

\begin{tikzpicture}
	[
	grow                    = right,
	sibling distance        = 6em,
	level distance          = 6em,
	edge from parent/.style = {draw, -latex},
	every node/.style       = {font=\footnotesize},
	sloped
	]
	\node [root] {$\gamma\vdash (\phi\lol\psi)\ot\rho$}
	[child anchor=west]
	child { node [env] {$\phi\vdash\psi$\\$\gamma\vdash\rho$}
		edge from parent node [below] {} }
	child { node [env] {$\psi\vdash\phi$\\$\gamma\vdash (\phi\lol\psi)\ot\rho$}
		child { node [env] {$\vdash\lnot\psi$\\$\psi\vdash\phi$\\$\gamma\vdash (\phi\lol\psi)\ot\rho$}
			child { node [env] {$\vdash\lnot\phi$\\$\vdash\lnot\psi$\\$\gamma\vdash\rho$}
				edge from parent node [below] {} }
			child { node [env] {$\vdash\lnot\lnot\phi$\\$\vdash\lnot\psi$\\$\vdash\lnot\gamma$}
				edge from parent node [above] {}
			}
			edge from parent node [below] {} }
		child { node [env] {$\vdash\lnot\lnot\psi$\\$\psi\vdash\phi$\\$\gamma\ot \phi\vdash \psi\ot\rho$}
			edge from parent node [above, align=center]
			{}
		}
		edge from parent node [above] {} };
\end{tikzpicture}

Rule 6:\smallskip

\begin{tikzpicture}
	[
	grow                    = right,
	sibling distance        = 3em,
	level distance          = 10em,
	edge from parent/.style = {draw, -latex},
	every node/.style       = {font=\footnotesize},
	sloped
	]
	\node [root] {$(\phi\lor\psi)\ot\rho\vdash\gamma$}
	child { node [env] {$\phi\ot\rho\vdash\gamma$ \\ $\psi\ot\rho\vdash\gamma$}
		edge from parent node [above] {}};
\end{tikzpicture}

Rule 7:\smallskip

\begin{tikzpicture}
	[
	grow                    = right,
	sibling distance        = 3em,
	level distance          = 10em,
	edge from parent/.style = {draw, -latex},
	every node/.style       = {font=\footnotesize},
	sloped
	]
	\node [root] {$(\phi\land\psi)\ot\rho\vdash\gamma$}
	child { node [env] {$\phi\vdash\psi$\\$\phi\ot\rho\vdash\gamma$}
		edge from parent node [below] {} }
	child { node [env] {$\psi\vdash\phi$\\$\psi\ot\rho\vdash\gamma$}
		edge from parent node [above] {} };
\end{tikzpicture}

Rule 8: For $r>0$ \smallskip

\begin{tikzpicture}
	[
	grow                    = right,
	sibling distance        = 3em,
	level distance          = 10em,
	edge from parent/.style = {draw, -latex},
	every node/.style       = {font=\footnotesize},
	sloped
	]
	\node [root] {$r*\bot\ot\rho\vdash\gamma$}
	child { node [env] {$\vdash\top$}
		edge from parent node [above] {}};
\end{tikzpicture}

Rule 9:\smallskip

\begin{tikzpicture}
	[
	grow                    = right,
	sibling distance        = 6em,
	level distance          = 6em,
	edge from parent/.style = {draw, -latex},
	every node/.style       = {font=\footnotesize},
	sloped
	]
	\node [root] {$(\phi\lol\psi)\ot\rho\vdash\gamma$}
	[child anchor=west]
	child { node [env] {$\phi\vdash\psi$\\$\rho\vdash\gamma$}
		edge from parent node [below] {} }
	child { node [env] {$\psi\vdash\phi$\\$(\phi\lol\psi)\ot\rho\vdash\gamma$}
		child { node [env] {$\vdash\lnot\psi$\\$\psi\vdash\phi$\\$(\phi\lol\psi)\ot\rho\vdash\gamma$}
			child { node [env] {$\vdash\lnot\phi$\\$\vdash\lnot\psi$\\$\rho\vdash\gamma$}
				edge from parent node [below] {} }
			child { node [env] {$\vdash\lnot\lnot\phi$\\$\vdash\lnot\psi$}
				edge from parent node [above] {}
			}
			edge from parent node [below] {} }
		child { node [env] {$\vdash\lnot\lnot\psi$\\$\psi\ot\rho\vdash\gamma\ot \phi$}
			edge from parent node [above, align=center]
			{}
		}
		edge from parent node [above] {} };
\end{tikzpicture}

Rule 10: For $r>0$,  \smallskip

\begin{tikzpicture}
	[
	grow                    = right,
	sibling distance        = 3em,
	level distance          = 10em,
	edge from parent/.style = {draw, -latex},
	every node/.style       = {font=\footnotesize},
	sloped
	]
	\node [root] {$\vdash\lnot\lnot r*\bot$}
	child { node [env] {$\top\vdash\bot$}
		edge from parent node [above] {}};
\end{tikzpicture}

Rule 11: For $r>0$,\smallskip

\begin{tikzpicture}
	[
	grow                    = right,
	sibling distance        = 3em,
	level distance          = 10em,
	edge from parent/.style = {draw, -latex},
	every node/.style       = {font=\footnotesize},
	sloped
	]
	\node [root] {$\vdash\lnot\lnot r*\top$}
	child { node [env] {$\vdash\top$}
		edge from parent node [above] {}};
\end{tikzpicture}

Rule 12: For $r>0$,\smallskip

\begin{tikzpicture}
	[
	grow                    = right,
	sibling distance        = 3em,
	level distance          = 10em,
	edge from parent/.style = {draw, -latex},
	every node/.style       = {font=\footnotesize},
	sloped
	]
	\node [root] {$\vdash\lnot\lnot r*\un$}
	child { node [env] {$\vdash\top$}
		edge from parent node [above] {}};
\end{tikzpicture}

Rule 13:\smallskip

\begin{tikzpicture}
	[
	grow                    = right,
	sibling distance        = 3em,
	level distance          = 10em,
	edge from parent/.style = {draw, -latex},
	every node/.style       = {font=\footnotesize},
	sloped
	]
	\node [root] {$\vdash\lnot\lnot (\phi\land\psi)$}
	child { node [env] {$\vdash\lnot\lnot \phi$\\$\vdash\lnot\lnot\psi$}
		edge from parent node [above] {}};
\end{tikzpicture}

Rule 14:\smallskip

\begin{tikzpicture}
	[
	grow                    = right,
	sibling distance        = 3em,
	level distance          = 10em,
	edge from parent/.style = {draw, -latex},
	every node/.style       = {font=\footnotesize},
	sloped
	]
	\node [root] {$\vdash\lnot\lnot (\phi\land\psi)$}
	child { node [env] {$\phi\vdash\psi$\\$\vdash\lnot\lnot \psi$}
		edge from parent node [below] {} }
	child { node [env] {$\psi\vdash\phi$\\$\vdash\lnot\lnot \phi$}
		edge from parent node [above] {} };
\end{tikzpicture}

Rule 15:\smallskip 

\begin{tikzpicture}
	[
	grow                    = right,
	sibling distance        = 3em,
	level distance          = 10em,
	edge from parent/.style = {draw, -latex},
	every node/.style       = {font=\footnotesize},
	sloped
	]
	\node [root] {$\vdash\lnot\lnot (\phi\ot\psi)$}
	child { node [env] {$\vdash\lnot\lnot \phi$\\$\vdash\lnot\lnot \psi$}
		edge from parent node [above] {}};
\end{tikzpicture}

Rule 16: \smallskip

\begin{tikzpicture}
	[
	grow                    = right,
	sibling distance        = 3em,
	level distance          = 10em,
	edge from parent/.style = {draw, -latex},
	every node/.style       = {font=\footnotesize},
	sloped
	]
	\node [root] {$\vdash\lnot\lnot (\phi\lol\psi)$}
	child { node [env] {$\phi\vdash\psi$}
		edge from parent node [below] {} }
	child { node [env] {$\psi\vdash\phi$\\$\vdash\lnot\lnot\psi$}
		edge from parent node [above] {} };
\end{tikzpicture}

The function $\D$ applies uniformly to $\LL$, $\LLun$ and $\LLs$. 
The only difference is that in $\LL$ we don't use Rule 12, and in 
$\LL$ and $\LLun$ the value of $r>0$ in the Rules 4, 8, 10, 11, and 12 
will be an integer and $r*\phi$ is in fact $r\phi$.

\medskip

The following results does not depend on the nondeterministic 
choices used in computing $\D$.
\begin{lemma}\label{discrimination1}
Let $\gamma$ be a judgement in $\mathcal{L} \in\{\LL,\LLun,\LLs\}$.
\begin{enumerate}
\item Any disjunctive theory in $\mathcal{L}$ that contains $\gamma$ 
also contains all the judgements of at least one of the elements in $\D(\gamma)$.

\item Any disjunctive theory in $\mathcal{L}$ that contains the judgements 
of at least one element in $\D(\gamma)$ also contains $\gamma$.

\item $\M$ is a model for $\gamma$ iff $\M$ is a model for at least 
one of the elements in $\D(\gamma)$.
\end{enumerate}
\end{lemma} 

\begin{proof}
Suppose that the signature of $\D(\gamma)$ is $\{F_1,\dots,F_n\}$.
	 
	1. Walking from the root to the leaves of the decision tree $\D(\gamma)$, one needs to chose between  supplementary judgements every time a node splits. Given a disjunctive theory $\T$ that contains $\gamma$, it must contain a set of choices of these  supplementary judgements; let $S$ be it. Consequently, this is (or contains) the trace of the path from the root to a particular leaf $F$. Hence, there exists a set $V$ of judgements so that $F=V\cup S$. Moreover, $S\subseteq\T$ and $\gamma\in \T$, meaning $$\dfrac{\T}{S~~~\gamma}.$$ From the corollary \ref{dectree} we know that $$\dfrac{S~~~\gamma}{F},~~~\text{ hence, }~~~\dfrac{\T}{F}.$$
	
	2. Let $F$ be the label of a leaf, and suppose that $\T$ is a disjunctive theory such that $F\subseteq\T$. From the definition of decision tree we know that $\dfrac{F}{\gamma}.$ Hence $\gamma\in\T$.
	
	3. ($\Rightarrow$) Suppose $\gamma$ has a model $\M$. From lemma \ref{max}, there exists a disjunctive consistent theory $\T_\M$ containing $\gamma$ and $\M$ is a model of $\T$. Applying 1. from this lemma, at least one of the elements of $\D(\gamma)$ is a subset of $\T_\M$. Hence, this element has $\M$ as a model.
	
	($\Leftarrow$) Suppose that an element $V\in\D(\gamma)$ admits a model $\M$. From lemma \ref{max}, there exists a disjunctive consistent theory $\T_\M$ that contains all the judgements satisied by $\M$, hence, $V\subseteq\T_\M$. Applying 2. of this lemma, $\gamma\in\T_\M$, hence $\M$ is a model of $\gamma$.
\end{proof}

 
 Next, we extend $\D$ to take  as inputs finite sets of judgements. We do it as follows, where $V$ is an arbitrary set of judgements:
\begin{align*}
\D(\emptyset) &\coloneq \emptyset \\
 \D(V) &\coloneq \{(V\setminus\{\gamma\})\cup T\mid \gamma\in V \text{ and }T\in\D(\gamma)\} \,.
 \end{align*}
Observe that $\D$ remains a nondeterministic function and its action can still be presented as a decision tree. Moreover, if we apply $\D$ to a finite set of judgements $V$ we get a decision tree; and if we apply $\D$ again to the labels of its leafs we get a new decision tree; and we can continue applying $\D$ to the labels of the leafs of the previous tree, obtaining larger and larger trees. However, this process will eventually end when one cannot apply $\D$ any more. This happens when all the judgements of the leafs are in normal form. Let's denote by $\D^*(V)$ this decision tree. Of course $\D^*(V)$ is not uniquely defined and one can get differently shaped trees depending on the nondeterministic choices in applying $\D$. However, for the results we present hereafter it is not important which of these trees we choose.

Given a finite set $V$ of judgements and an atomic 
proposition $p\in\Prop$, the \textit{$p$-saturation of $V$} 
is defined as follows 
$$sat_p(V)=\{V\cup\{\vdash\lnot p\},V\cup\{\vdash\lnot\lnot p\}\}.$$ 
This function can be represented by the following decision tree:

\begin{tikzpicture}
	[
	grow                    = right,
	sibling distance        = 3em,
	level distance          = 10em,
	edge from parent/.style = {draw, -latex},
	every node/.style       = {font=\footnotesize},
	sloped
	]
	\node [root] {$V$}
	child { node [env] {$\vdash\lnot p$\\$V$}
		edge from parent node [below] {} }
	child { node [env] {$\vdash\lnot\lnot p$\\$V$}
		edge from parent node [above] {} };
\end{tikzpicture}

Let $sat(V)$ be the composition of the decision trees of $sat_p(V)$ for all atomic propositions $p$ that appear in at least one judgement in $V$. Their order is not important. If $P$ is the set of atomic propositions that appear in the judgements of $V$, then 
\begin{align*}
sat(V) ={} &\{V\cup\{\vdash\lnot x \mid x\in W\} \cup {} \\
	&\{\vdash\lnot\lnot y\mid y\in P\setminus W\}\mid W\subseteq P\}.
\end{align*}

If $V$ is a finite set of judgements, let the \textit{refinement of $V$} be the set $\mathit{ref}(V)$ of judgements obtained as follows:
\begin{enumerate}
	\item if either $\vdash\bot\in V$, or for some formula $\phi$ both $\vdash\lnot\phi$, $\vdash\lnot\lnot\phi\in V$, replace $V$ with $\{\vdash\bot\}$;
	\item identify all alethic judgements in $V$ and do simultaneously
	\begin{itemize}
		\item if $\top\vdash p\in V$ replace all the occurrences of $p$ in all the other judgements in $V$ with $\top$;
		\item if $p\vdash\bot\in V$ replace all the occurrences of $p$ in all the other judgements in $V$ with $\bot$;
	\end{itemize}
	\item let $\mathit{ref}(V)$ be the result of there replacements.		
\end{enumerate}	

These replacements produce equivalent results in every Lawvere logic and for this reason we can describe the effect of refinement using the following decision tree:

\begin{tikzpicture}
	[
	grow                    = right,
	sibling distance        = 3em,
	level distance          = 10em,
	edge from parent/.style = {draw, -latex},
	every node/.style       = {font=\footnotesize},
	sloped
	]
	\node [root] {$V$}
	child { node [env] {$\mathit{ref}(V)$}
		edge from parent node [above] {}};
\end{tikzpicture}

Next, we use these functions to present an algorithm that computes the normal representation of a finitely axiomatized theory (identified by its finite axiomatization). Moreover, being the presentation of these functions as decision trees, the algorithm itself is represented by a decision tree. 

\subsection*{The Normalisation Algorithm}

\begin{enumerate}
	\item \textbf{Input:} a finite set $R$ of judgements;
	\item\textbf{Let} $\mathcal X:=sat(R)$;
	\item \textbf{For each} leaf identity $F$ of $\mathcal X$ \textbf{compute} $\D^*(F)$; \\ \textbf{let} $\mathcal Y$ be the decision tree obtained by composing $\mathcal X$ with these trees;
	\item \textbf{For each} leaf identity $W$ of $\mathcal Y$ \textbf{compute} $\mathit{ref}(W)$; \\ \textbf{let} $\mathcal Z$ be the decision tree obtained by composing $\mathcal Y$ with these trees;
	\item \textbf{If} $\mathcal Z\neq\mathcal X$, \textbf{let} $\mathcal X:=\mathcal Z$ and \textbf{go back to 3};
	\item \textbf{Else, output} $\mathcal Z$.
\end{enumerate}

Observe that the algorithm always terminates on finite inputs. Since it uses $\D$, it is nondeterministic. However, the results presented hereafter remain true independently of the nondeterministic choices. Also, the computation of the algorithm can be represented as a decision tree with the root indexed by $R$, and the structure given by the composition of the decision trees used in the steps of the algorithm. 
%
%

\end{document}

\section{Useful lemmas}
In this subsection, we prove a series of theorems of \LLQ\ as well as some useful derived rules. For $\gamma_i$, $\gamma$ judgements, a \emph{derived rule}
\begin{equation*}
	\infrule{\gamma_1 & \cdots & \gamma_n}{\gamma}
\end{equation*}
denotes that $\gamma$ is provable from $\{ \gamma_1, \dots, \gamma_n \}$. By denoting this provability statements in the form of rules of inference will just make easier their applications in the formation of new proofs, as by adding them in the
deductive systems will not change the derivable theories.

\begin{lemma}\label{struct-lemma}
	The following are provable in all \LLQ:
	\begin{align*} 
		\textbf{1.} \quad 
		&\infrule{\Gamma, \phi\land\psi\vdash\theta & \phi\vdash\psi
		}{\Gamma,\phi\vdash\theta} &  
		\textbf{2.} \quad
		&\infrule{\Gamma\vdash\phi\lor\psi & \phi\vdash\psi}{\Gamma\vdash\psi} 
		\\[1ex]
		\textbf{3.} \quad 
		&\infrule{\Gamma\vdash\phi & \Delta\vdash\psi
		}{\Gamma,\Delta\vdash\phi\land\psi} &
		\textbf{4.} \quad
		&\infrule{\Gamma,\top\vdash\phi}{\Gamma\vdash\phi} 
		\\[1ex]
		\textbf{5.} \quad
		&\infrule{}{\phi\vdash\phi\lor\psi} \quad & 
		\textbf{6.} \quad
		&\infrule{}{\phi\land\psi\vdash\phi} 
		\\[1ex]
		\textbf{7.} \quad
		&\infrule{}{\vdash\phi\land\psi \lollol \psi\land\phi}  & 
		\textbf{8.} \quad 
		&\infrule{}{\vdash\phi\lor\psi \lollol \psi\lor\phi} 
		\\[1ex]
		\textbf{9.} \quad
		&\infrule{}{\vdash(\phi\land\psi)\lor\phi \lollol \phi} & 
		\textbf{10.} \quad 
		&\infrule{}{\vdash(\phi\lor\psi)\land\phi \lollol \phi}
		\\[1ex]
		\textbf{11.} \quad
		&\infrule{}{\vdash\phi\land\top \lollol \phi} & 
		\textbf{12.} \quad 
		&\infrule{}{\vdash\phi\lor\top \lollol \top}
		\\[1ex]
		\textbf{13.} \quad
		&\infrule{}{\vdash\phi\land\bot \lollol \bot} & 
		\textbf{14.} \quad 
		&\infrule{}{\vdash\phi\lor\bot \lollol \phi}
	\end{align*}
\end{lemma}

\begin{proof}
	We prove a few cases. The other cases are as easy.
	
	1. 
	\begin{equation*}
		\infrule[cut]{
			\infrule[$\land_2$]{
				\phi\vdash\phi 
				& 
				\phi\vdash\psi
			}{\phi\vdash\phi\land\psi} 
			& 
			\Gamma,\phi\land\psi\vdash\theta
		}{\Gamma,\phi\vdash\theta}
	\end{equation*}
	
	3.
	\begin{equation*}
		\infrule[$\land_2$]{
			\infrule[weak]{\Gamma\vdash\phi}{\Gamma, \Delta\vdash\phi} 
			& 
			\infrule[weak]{\Delta\vdash\psi}{\Gamma,\Delta\vdash\psi}
		}{\Gamma,\Delta\vdash\phi\land\psi}
	\end{equation*}
	
	4.
	\begin{equation*}
		\infrule[cut]{
			\infrule[top]{}{\vdash\top}
			& 
			\Gamma,\top\vdash\phi
		}{\Gamma\vdash\phi}
	\end{equation*}
\end{proof}

\begin{lemma}\label{quantale-lemma}
	The following are provable in all \LLQ:
	\begin{align*}
		&\textbf{1.} \quad \infrule{}{\vdash \phi\ot(\psi\ot\rho) \lollol (\phi\ot\psi)\ot\rho}
		\\[1ex]
		&\begin{aligned}
			\textbf{2.} \quad 
			& \infrule{\Gamma\vdash\phi\ot\psi}{\Gamma\vdash\phi} &  
			\textbf{3.} \quad 
			& \infrule{\Gamma\vdash\phi & \Delta\vdash\psi}{\Gamma,\Delta\vdash\phi\ot\psi} 
			\\[1ex]
			\textbf{4.} \quad 
			& \infrule{}{\phi\ot\psi\vdash\phi} &
			\textbf{5.} \quad 
			& \infrule{}{\vdash \phi\ot\psi\lollol\psi\ot\phi}
			\\[1ex]
		\textbf{6.} \quad
		& \infrule{}{\vdash \phi\ot\top\lollol\phi } & 
		\textbf{7.} \quad 
		& \infrule{}{\vdash \phi\ot\bot\lollol\bot }
		\\[1ex]
		\textbf{8.} \quad
		& \doubleinfrule{\vdash\phi\lol\psi}{\phi\vdash\psi} & 
		\textbf{9.} \quad 
		& \infrule{\vdash\psi}{\vdash\phi\lol\psi}
		\\[1ex]
		\textbf{10.} \quad
		& \infrule{}{\vdash\phi\lol\top} & 
		\textbf{11.} \quad 
		& \infrule{}{\vdash\phi\lol\phi }
		\\[1ex]
		\textbf{12.} \quad
		& \infrule{}{\phi\vdash\psi\lol\phi } & 
		\textbf{13.} \quad 
		& \infrule{}{\phi\ot(\phi\lol\psi)\vdash\psi }
		\\[1ex]
		\textbf{14.} \quad
		& \infrule{\psi\vdash\phi}{\psi\vdash\phi\ot(\phi\lol\psi)} & 
		\textbf{15.} \quad 
		& \infrule{\phi\vdash\psi & \theta\vdash\rho}{\phi\ot\theta\vdash\psi\ot\rho}
	\end{aligned}
\end{align*}
\end{lemma}

\begin{proof}
We prove a few cases. The other cases are as easy.

1. We show the proof of one of the two linear implications. 
\begin{equation*}
	\infrule[$\ot_2$]{
		\infrule[$\ot_1$]{
			\infrule[$\ot_1$]{
				\infrule[Lemma~\ref{quantale-lemma}.(3)]{
					\infrule[Lemma~\ref{quantale-lemma}.(3)]{
						\phi\vdash\phi~~~\psi\vdash\psi}{\phi,\psi\vdash\phi\ot\psi}
					& 
					\rho\vdash\rho
				}{\phi,\psi,\rho\vdash(\phi\ot\psi)\ot\rho}
			}{\phi,\psi\ot\rho\vdash(\phi\ot\psi)\ot\rho}
		}{\phi\ot(\psi\ot\rho)\vdash(\phi\ot\psi)\ot\rho}
	}{\vdash \phi\ot(\psi\ot\rho) \lol (\phi\ot\psi)\ot\rho}
\end{equation*}
The other implication follows similarly. We conclude by ($\land_2$).

2. 
\begin{equation*}
	\infrule[cut]{
		\Gamma\vdash\phi\ot\psi
		&
		\infrule[$\ot_1$]{
			\infrule[weak]{\phi\vdash\phi}{\phi,\psi\vdash\phi}
		}{\phi\ot\psi\vdash\phi} 
	}{\Gamma\vdash\phi}
\end{equation*}

3. Let $\theta$ be the result of taking all the elements of $\Gamma$ 
and connect them by $\ot$ ---recall that $\ot$ is associative.
\begin{equation*}
	\infrule[rep. $\ot_1$]{
		\infrule[$\lol_1$]{
			\infrule[weak]{\Delta\vdash\psi}{\Delta,\phi\lol\theta\vdash\psi}
			&
			\infrule[rep. $\ot_1$]{\Gamma\vdash\phi}{\theta\vdash\phi}
		}{\theta,\Delta\vdash\phi\ot\psi}
	}{\Gamma,\Delta\vdash\phi\ot\psi}
\end{equation*}
\end{proof}

\begin{lemma} The following are provable in $\LLun$ and $\LLs$.
\begin{align*} 
	\text{\textbf{1.}} \;
	& \frac{\vdash\un}{\top\vdash\bot} \quad & 
	\text{\textbf{2.}} \;
	& \frac{\un\vdash\bot}{\top\vdash\bot} \, &
	\text{\textbf{3.}} \;
	& \frac{\phi\vdash\phi\ot\un}{\phi\vdash\bot} \quad & 
	\text{\textbf{4.}} \;
	& \frac{\un\lol\phi\vdash\phi}{\phi\vdash\bot} \,
\end{align*}
\end{lemma}

\begin{lemma} The following are provable in $\LLs$:
\begin{align*} 
	\textbf{1.} \quad
	& \infrule{}{(r+s)*\phi\vdash r*\phi} & 
	\textbf{2.} \quad
	& \infrule{}{r*\top\dashv\vdash\top }
	\\[1ex]
	\textbf{3.} \quad 
	& \frac{r*\un\vdash(r+s)*\un}{\top\vdash\bot}~s>0 \,
\end{align*}
\end{lemma}
%
%

\begin{lemma}\label{totality-lemma}
The following is provable in \LLQ.
\begin{equation*}
	\infrule{
		\Gamma,\phi\lol\psi\vdash\theta  
		& 
		\Gamma,\psi\lol\phi\vdash\theta
	}{\Gamma\vdash\theta}
\end{equation*}
\end{lemma}

\begin{proof}
\begin{equation*}
	\infrule[Lemma~\ref{struct-lemma}.(4)]{
		\infrule[cut + tot]{
			\infrule{
				\Gamma,\phi\lol\psi\vdash\theta
				&
				\Gamma,\psi\lol\phi\vdash\theta
			}{\Gamma,(\phi\lol\psi)\lor(\psi\lol\phi)\vdash\theta}
		}{\Gamma,\top\vdash\theta}
	}{\Gamma\vdash\theta}
\end{equation*}
\end{proof}

\begin{lemma}\label{meta1}
In all \LLQ, 
\begin{equation*}
	\text{If} \quad 
	\dfrac{\vdash\phi}{\vdash\psi} 
	\quad \text{then} \quad
	\left[ \dfrac{\vdash\phi\land\theta}{\vdash\psi\land\theta} 
	\quad\text{and}\quad
	\dfrac{\vdash\phi\lor\theta}{\vdash\psi\lor\theta}
	\right] \,.
\end{equation*}
\end{lemma}

\begin{proof}
We show one case.
\begin{equation*}
	\infrule[$\land_2$]{
		\infrule{
			\infrule[$\land_3^*$]{\vdash \phi \land \theta}{\vdash \phi}
		}{\vdash \psi}
		&
		\infrule[$\land_3$]{\vdash \phi \land \theta}{\vdash \theta}
	}{\vdash \psi \land \theta}
\end{equation*} 
where ($\land_3^*$) is the symmetric rule of ($\land_3$), which is derivable.
The other case follows similarly by using ($\lor_1$), ($\lor_2$), ($\lor_2^*$).
\end{proof}

\begin{lemma}\label{meta2}
In all \LLQ
\begin{equation*}
	\dfrac{\vdash\theta}{\vdash\phi} 
	\quad \text{ iff } \quad
	\dfrac{\vdash\theta}{\vdash\phi\land\theta}
\end{equation*}
\end{lemma}
\begin{proof}
($\Leftarrow$): $$\dfrac{\dfrac{\vdash\theta}{\vdash\phi\land\theta}}{\vdash\phi}.$$
($\Rightarrow$): $\dfrac{\vdash\theta}{\vdash\phi}$ implies, using lemma \ref{meta1}, $\dfrac{\theta\land\theta}{\theta\land\phi}$, which is equivalent to $\dfrac{\vdash\theta}{\vdash\phi\land\theta}$. 
\end{proof}

\begin{lemma}\label{meta3}
In all LLQ ,
$$\left[ ~\dfrac{\vdash\theta}{\vdash\phi}~~\text{ and }~~\dfrac{\vdash\theta}{\vdash\psi}~\right]~~~\text{ iff }~~~\dfrac{\vdash\theta}{\vdash\phi\land\psi}.$$
\end{lemma}

\begin{proof}
($\Leftarrow$): $$\dfrac{\dfrac{\vdash\theta}{\vdash\phi\land\psi}}{\vdash\theta}.$$
($\Rightarrow$): Using lemma \ref{meta1}, from $\dfrac{\vdash\theta}{\vdash\phi}$ we infer $\dfrac{\vdash\theta\land\psi}{\vdash\phi\land\psi}$. Similarly, applying lemma \ref{meta2}, from $\dfrac{\vdash\theta}{\vdash\psi}$ we derive $\dfrac{\vdash\theta}{\vdash\theta\land\psi}$. Next, modus ponens completes the proof. 
\end{proof}

\begin{lemma}\label{meta4}
In all LLQ ,
$$\text{If }~~\left[~\dfrac{\vdash\phi}{\vdash\theta}~~~\text{ and }~~~\dfrac{\vdash\psi}{\vdash\theta}\right]~~~~\text{ then }~~~~\dfrac{\vdash\phi\lor\psi}{\vdash\theta}.$$	
\end{lemma}

\begin{proof}
Applying lemma \ref{meta1} to $~~\dfrac{\vdash\phi}{\vdash\theta}~~$ and to $~~\dfrac{\vdash\psi}{\vdash\theta}~~$ we obtain $~~\dfrac{\vdash\phi\lor\psi}{\vdash\theta\lor\psi}~~$ and $~~\dfrac{\vdash\phi\lor\psi}{\vdash\theta\lor\phi}~~$ respectively. Using them in the context of lemma \ref{meta3} we get
$$\dfrac{\vdash\phi\lor\psi}{\vdash(\theta\lor\phi)\land(\theta\lor\psi)},$$ which implies further $$\dfrac{\vdash\phi\lor\psi}{\vdash\theta\land(\phi\lor\psi)}.$$ Now, applying Lemma~\ref{meta3}, we get 
\begin{equation*}
	\dfrac{\vdash\phi\lor\psi}{\vdash\theta} \,.
\end{equation*}
\end{proof}

The first totality lemma allows us to prove the distributivity of the tensor product with respect to conjunction and disjunction respectively.

\begin{lemma}\label{quantale-distrib-lemma}
	The following are provable in all LLQ .
	\begin{align*} 
		\text{\textbf{1.}} \quad
		& (\rho\ot\phi)\land(\rho\ot\psi)\dashv\vdash\rho\ot(\phi\land\psi) \quad 
		\\ \text{\textbf{2.}} \quad 
		& (\rho\ot\phi)\lor(\rho\ot\psi)\dashv\vdash\rho\ot(\phi\lor\psi)\,\\[1ex]
		\text{\textbf{3.}} \quad
		& \frac{\Gamma\vdash\rho\ot\phi~~~\Gamma\vdash\rho\ot\psi}{\Gamma\vdash\rho\ot(\phi\land\psi)} \quad 
		\\[1ex] \text{\textbf{4.}} \quad 
		& \frac{\Gamma\vdash\rho\ot\phi}{\Gamma\vdash\rho\ot(\phi\lor\psi)} \,;
	\end{align*}
\end{lemma}

\begin{proof}
	1. ($\Rightarrow$): We have
	$$\dfrac{\dfrac{\dfrac{\dfrac{\vdash\phi\lol\psi}{\phi\vdash\psi}}{\phi\vdash\phi\land\psi}}{\rho\ot\phi\vdash\rho\ot(\phi\land\psi)}}{(\rho\ot\phi)\land(\rho\ot\psi)\vdash\rho\ot(\phi\land\psi)}$$
	and similarly
	$$\dfrac{\dfrac{\dfrac{\dfrac{\vdash\psi\lol\phi}{\psi\vdash\phi}}{\psi\vdash\phi\land\psi}}{\rho\ot\psi\vdash\rho\ot(\phi\land\psi)}}{(\rho\ot\phi)\land(\rho\ot\psi)\vdash\rho\ot(\phi\land\psi)}$$
	Applying Totality Lemma (lemma \ref{totalityl1} proven below), we get $$(\rho\ot\phi)\land(\rho\ot\psi)\vdash\rho\ot(\phi\land\psi).$$
	($\Leftarrow$):
	$$\dfrac{\dfrac{\phi\land\psi\vdash\phi}{\rho\ot(\phi\land\psi)\vdash\rho\ot\phi}~~~~\dfrac{\phi\land\psi\vdash\psi}{\rho\ot(\phi\land\psi)\vdash\rho\ot\psi}}{\rho\ot(\phi\land\psi)\vdash(\rho\ot\phi)\land(\rho\ot\psi)}.$$
	2. is proven similarly.
\end{proof}